\documentclass{ieeetmlcn}
\usepackage{paper}
\tikzset{
  box/.style={draw, rounded corners=2pt, minimum height=0.5cm, inner sep=2.5pt, align=center, font=\scriptsize\bfseries},
  encoderbox/.style={box, fill=orange!30, draw=orange!70!black},
  decoderbox/.style={box, fill=green!25, draw=green!55!black},
  normbox/.style={box, fill=yellow!40, draw=yellow!45!black},
  lbl/.style={align=center, inner sep=1pt},
  arr/.style={-{Latex[length=1.6mm]}, line width=0.7pt},
}
\newcommand{\drawSystemModel}[1]{%
  \begin{tikzpicture}[
      scale=#1, transform shape,
      font=\scriptsize,
      pics/ant/.style={code={
        \draw[black!85, line width=0.7pt] (0,-0.26) -- (0,0.10);
        \draw[black!85, line width=0.7pt] (-0.14,0.30) -- (0,0.10) -- (0.14,0.30);
        \draw[black!85, line width=0.7pt] (-0.16,-0.26) -- (0.16,-0.26);
      }},
      pics/rf/.style={code={
        \foreach \r in {0.28,0.42,0.56}
          \draw[blue!55!black, line width=0.6pt] (-38:\r) arc (-38:38:\r);
      }},
    ]

    \node[lbl] (input) at (-3.95, 0) {Input\\[-1pt] $x$};
    \node[encoderbox] (enc) at (-2.9, 0) {Encoder\\[-1pt] $\mathcal{E}$};
    \node[normbox] (pn) at (-1.35, 0) {Power\\[-1pt] Norm.};
    \node[decoderbox] (dec) at (2.25, 0) {Decoder\\[-1pt] $\mathcal{D}$};
    \node[lbl] (output) at (3.4, 0) {Output\\[-1pt] $\hat{x}$};

    \pic at (-0.15, 0) {ant};
    \pic at (-0.15, 0.10) {rf};
    \pic at (1.05, 0) {ant};

    \draw[arr] (input) -- (enc);
    \draw[arr] (enc) -- (pn);
    \draw[arr] (pn.east) -- (-0.38, 0) node[midway, above, lbl] {$s$};
    \draw[arr] (1.28, 0) -- (dec.west) node[midway, above, lbl] {$\hat{s}$};
    \draw[arr] (dec) -- (output);

    \node[lbl] at (-0.15, -0.46) {TX};
    \node[lbl] at (1.05, -0.46) {RX};
    \node[lbl] at (0.45, -0.78) {ADALM-PlutoSDR, over the air};

  \end{tikzpicture}%
}

\newcommand{\drawIbpBox}[1]{%
  \begin{tikzpicture}[
      scale=#1, transform shape,
      font=\scriptsize,
      >=Latex,
    ]
    \draw[->] (-1.6,0) -- (1.8,0) node[right] {$z$};
    \draw[->] (0,-0.55) -- (0,1.7) node[above] {$\hat{z}$};

    \fill[blue!15, opacity=0.45] (-1.2,0) rectangle (1.2,1.2);
    \draw[thick, blue!70!black] (-1.2,0) rectangle (1.2,1.2);

    \draw[thick, black] (-1.4,0) -- (0,0) -- (1.4,1.4);

    \coordinate (L) at (-1.2,0);
    \coordinate (U) at (1.2,1.2);
    \fill (L) circle (1.2pt);
    \node[above left, inner sep=1.5pt] at (-1.25,0.02) {$(\underline{z},0)$};
    \fill (U) circle (1.2pt);
    \node[right] at (1.42,1.28) {$(\overline{z},\overline{z})$};
    \node[below] at (-1.2,-0.05) {$\underline{z}$};
    \node[below] at (1.2,-0.05) {$\overline{z}$};
    \node[above left, inner sep=2pt] at (-0.05,1.2) {$\overline{\hat{z}}$};
    \node[below left, inner sep=2pt] at (0,0) {$\underline{\hat{z}}{=}0$};
  \end{tikzpicture}%
}

\newcommand{\drawReluWedge}[1]{%
  \begin{tikzpicture}[
      scale=#1, transform shape,
      font=\scriptsize,
      >=Latex,
    ]
    \draw[->] (-1.6,0) -- (1.8,0) node[right] {$z$};
    \draw[->] (0,-0.65) -- (0,1.7) node[above] {$\hat{z}$};

    \fill[blue!15, opacity=0.45]
      (-1.2,-0.42) -- (1.2,0.42) -- (1.2,1.2) -- (-1.2,0) -- cycle;

    \draw[thick, black] (-1.4,0) -- (0,0) -- (1.4,1.4);

    \coordinate (L) at (-1.2,0);
    \coordinate (U) at (1.2,1.2);
    \draw[dotted, black!55, line width=0.6pt] (U) -- (1.2,0);
    \draw[dotted, black!55, line width=0.6pt] (U) -- (0,1.2);
    \draw[dotted, black!55, line width=0.6pt] (L) -- (-1.2,-0.42);
    \fill (L) circle (1.2pt);
    \node[above left, inner sep=1.5pt] at (-1.25,0.02) {$(\underline{z},0)$};
    \fill (U) circle (1.2pt);
    \node[right] at (1.42,1.28) {$(\overline{z},\overline{z})$};
    \node[below right, inner sep=1.5pt] at (-1.2,-0.02) {$\underline{z}$};
    \node[below] at (1.2,-0.05) {$\overline{z}$};
    \node[left] at (0,1.2) {$\overline{z}$};

    \draw[thick, blue!70!black]
      ($(L)!-0.08!(U)$) -- ($(U)!-0.12!(L)$)
      node[midway, above left, blue!70!black] {upper};

    \draw[dashed, red!45, line width=0.5pt] (-1.4,0) -- (1.4,0);
    \node[below right, red!70!black, inner sep=2pt] at (1.35,-0.02) {$\alpha{=}0$};
    \draw[dashed, orange!80!black!60, line width=0.5pt] (-0.5,-0.5) -- (1.6,1.6)
      node[right, orange!80!black] {$\alpha{=}1$};

    \draw[very thick, dashed, red!70!black] (-1.4,-0.49) -- (1.6,0.56)
      node[right, red!70!black] {lower};
    \draw[{Stealth[length=2.5pt]}-{Stealth[length=2.5pt]}, red!70!black, line width=0.6pt]
      (0:0.6) arc[start angle=0, end angle=45, radius=0.6];
    \node[below, red!70!black, inner sep=2pt] at (0.52,-0.02) {$\alpha{\in}[0,1]$};
  \end{tikzpicture}%
}

\newcommand{\drawPreluWedge}[1]{%
  \begin{tikzpicture}[
      scale=#1, transform shape,
      font=\scriptsize,
      >=Latex,
    ]
    \draw[->] (-1.6,0) -- (1.8,0) node[right] {$z$};
    \draw[->] (0,-0.75) -- (0,1.7) node[above] {$\hat{z}$};

    \fill[blue!15, opacity=0.45]
      (-1.2,-0.54) -- (1.2,0.54) -- (1.2,1.2) -- (-1.2,-0.3) -- cycle;

    \draw[thick, black] (-1.4,-0.35) -- (0,0) -- (1.4,1.4);

    \coordinate (L) at (-1.2,-0.3);
    \coordinate (U) at (1.2,1.2);
    \draw[dotted, black!55, line width=0.6pt] (L) -- (-1.2,0);
    \draw[dotted, black!55, line width=0.6pt] (U) -- (1.2,0);
    \draw[dotted, black!55, line width=0.6pt] (U) -- (0,1.2);
    \draw[dotted, black!55, line width=0.6pt] (-1.2,-0.3) -- (-1.2,-0.54);
    \fill (L) circle (1.2pt);
    \node[above left, inner sep=1.5pt] at (-1.25,-0.3) {$(\underline{z},a\underline{z})$};
    \fill (U) circle (1.2pt);
    \node[right] at (1.42,1.28) {$(\overline{z},\overline{z})$};
    \node[above right, inner sep=1.5pt] at (-1.2,0.0) {$\underline{z}$};
    \node[below] at (1.2,-0.05) {$\overline{z}$};
    \node[left] at (0,1.2) {$\overline{z}$};

    \draw[thick, blue!70!black]
      ($(L)!-0.08!(U)$) -- ($(U)!-0.12!(L)$)
      node[midway, above left, blue!70!black] {upper};

    \draw[dashed, red!45, line width=0.5pt] (-1.4,-0.35) -- (1.4,0.35)
      node[right, red!70!black] {slope $a$};
    \draw[dashed, orange!80!black!60, line width=0.5pt] (-0.5,-0.5) -- (1.6,1.6)
      node[right, orange!80!black] {slope $1$};

    \draw[very thick, dashed, red!70!black] (-1.4,-0.63) -- (1.6,0.72)
      node[right, red!70!black] {lower};
    \draw[{Stealth[length=2.5pt]}-{Stealth[length=2.5pt]}, red!70!black, line width=0.6pt]
      (14:0.75) arc[start angle=14, end angle=45, radius=0.75];
    \node[below, red!70!black, inner sep=1pt] at (0.55,-0.3) {slope $\in[a,1]$};
  \end{tikzpicture}%
}

\tikzset{
  intervalbox/.style={draw, dashed, rounded corners=2pt, minimum height=0.5cm, inner sep=2.5pt, align=center, font=\scriptsize},
  safebox/.style={box, fill=green!20, draw=green!55!black},
  unknownbox/.style={box, fill=gray!20, draw=gray!55!black},
  decision/.style={draw, diamond, aspect=2.2, fill=blue!15, font=\scriptsize, inner sep=1pt, align=center},
}
\newcommand{\drawVerificationFramework}[1]{%
  \begin{tikzpicture}[
      scale=#1, transform shape,
      font=\scriptsize,
    ]

    \node[intervalbox] (phiin) at (-5.5, 0) {$\phi_{in}$\\[-1pt] on $\hat s$};
    \node[decoderbox] (dec) at (-4.0, 0) {Decoder\\[-1pt] $\mathcal{D}$};
    \node[intervalbox] (xhat) at (-2.3, 0) {$[\underline{\hat x},\overline{\hat x}]$};
    \node[lbl] (x) at (-0.75, 0.9) {clean image\\[-1pt] $x$};
    \node[normbox] (mae) at (-0.75, 0) {$\mathrm{MAE}$\\[-1pt] via~\eqref{eq:mae-relu-identity}};
    \node[intervalbox] (bound) at (0.9, 0) {$[\mathrm{lb},\mathrm{ub}]$};
    \node[decision] (dec2) at (2.7, 0) {$\mathrm{ub}\!\le\!\tau$?};
    \node[safebox] (safe) at (4.5, 0.75) {safe};
    \node[unknownbox] (unknown) at (4.5, -0.75) {unknown};

    \draw[arr] (phiin) -- (dec);
    \draw[arr] (dec) -- (xhat);
    \draw[arr] (xhat) -- (mae);
    \draw[arr] (x) -- (mae);
    \draw[arr] (mae) -- (bound);
    \draw[arr] (bound) -- (dec2);
    \draw[arr] (dec2) -- (safe) node[midway, above left, lbl] {yes};
    \draw[arr] (dec2) -- (unknown) node[midway, below left, lbl] {no};

  \end{tikzpicture}%
}

\newcommand{\drawFadingArc}[1]{%
  \begin{tikzpicture}[
      scale=#1, transform shape,
      font=\scriptsize,
      >=Latex,
    ]
    \def\deltadeg{40}
    \pgfmathsetmacro{\cdelta}{cos(\deltadeg)}
    \pgfmathsetmacro{\sdelta}{sin(\deltadeg)}

    \draw[->] (-1.55,0) -- (1.7,0) node[right] {$\mathrm{Re}(h)$};
    \draw[->] (0,-1.35) -- (0,1.55) node[above] {$\mathrm{Im}(h)$};

    \draw[dashed, black!40, line width=0.6pt] (0,0) circle (1);

    \fill[blue!15, opacity=0.45]
      (\cdelta,-\sdelta) rectangle (1,\sdelta);
    \draw[thick, blue!70!black]
      (\cdelta,-\sdelta) rectangle (1,\sdelta);

    \draw[very thick, black]
      (-\deltadeg:1) arc (-\deltadeg:\deltadeg:1);

    \draw[dotted, black!55, line width=0.6pt] (0,0) -- (\deltadeg:1.15);
    \draw[dotted, black!55, line width=0.6pt] (0,0) -- (-\deltadeg:1.15);

    \draw[black!70, line width=0.5pt]
      (0.28,0) arc (0:\deltadeg:0.28);
    \node[black!70] at (18:0.42) {$\delta$};
    \draw[black!70, line width=0.5pt]
      (0.28,0) arc (0:-\deltadeg:0.28);
    \node[black!70] at (-18:0.42) {$\delta$};

    \fill (1,0) circle (1.3pt);
    \node[below right, inner sep=1.5pt] at (1.02,-0.02) {$h{=}1$};
    \fill (\deltadeg:1) circle (1.3pt);
    \node[above right, inner sep=1pt] at (\deltadeg:1.02) {$e^{j\delta}$};
    \fill (-\deltadeg:1) circle (1.3pt);
    \node[below right, inner sep=1pt] at (-\deltadeg:1.02) {$e^{-j\delta}$};

  \end{tikzpicture}%
}

\settoggle{usecomment}{true}

\begin{document}

\receiveddate{XX Month, XXXX}
\reviseddate{XX Month, XXXX}
\accepteddate{XX Month, XXXX}
\publisheddate{XX Month, XXXX}
\currentdate{XX Month, XXXX}
\doiinfo{TMLCN.2026.1234567}

\markboth{}{Author \MakeLowercase{\textit{et al.}}}

\title{Neural Network Verification for Deep Joint Source-Channel Coding}

\author{
  Thanh Le\authorrefmark{1}, Hai Duong\authorrefmark{2}, Takeshi Matsumura\authorrefmark{1}, and ThanhVu Nguyen\authorrefmark{2}
}
\affil{Wireless Systems Laboratory, National Institute of Information and Communication Technology, Yokosuka, Kanagawa, Japan}
\affil{George Mason University, Fairfax, VA, USA}
\corresp{Corresponding author: Thanh Le (email: hi@thanhle.xyz).}

\begin{abstract}
  Deep joint source-channel coding (\djscc{}) transmits data end-to-end over wireless channels using a neural encoder-decoder,
  but reconstruction quality can degrade sharply under adversarial perturbations and channel disturbances; no method formally bounds this degradation for \djscc{}.
  We present the first bound-propagation framework for verifying \djscc{}'s decoder, bounding worst-case reconstruction error over a given wireless channel's noise region.
  Current deep neural network (DNN) verifiers do not support three \djscc{} decoder components: parametric rectified linear activations (PReLU), transposed convolutions, and Rayleigh fading.
  We extend state-of-the-art techniques for optimization of linear relaxation in DNN verification for PReLU,
  replace the transposed convolution with its restricted upsample-then-convolution form,
  and formulate Rayleigh fading as a structural perturbation prepended directly into the decoder,
  thereby reducing the dimensionality of the verification problem.
  We also instantiate Lipschitz-regularized global robustness training, denoted \gloro{}, improving global robustness and enabling tight certification of \djscc{} models for the first time.
  On \djscc{} model for image transmission, this global robustness training procedure combined with structural encoding lowers the median certified bound by up to 41\% and
  certifies about ten times more safe cases (192 against 19) than \gloro{} with interval encoding at a 10-degree error in channel estimation.
  Over-the-air validation with an orthogonal frequency-division multiplexing (OFDM) implementation on software-defined radio devices confirm the certificate holds on real hardware,
  with a worst observed error on radio link at 0.082 against a certified bound of 0.128.
\end{abstract}


\maketitle

\section{Introduction}
  \label{sec:intro}
  Neural networks can encode and decode wireless transmissions end-to-end,
  replacing hand-crafted source compression and channel codes with a single trained pipeline~\cite{o2017introduction,bourtsoulatze2019deep}.
  \djscc{}~\cite{bourtsoulatze2019deep} is the first image transmission model built this way:
  it trains an encoder-decoder pair jointly to map images to complex channel symbols,
  adapting to channel quality without a separate compression stage and without retraining for each new noise level.
  Transformer-based successors such as the Wireless Image Transmission Transformer (WITT)~\cite{yang2023witt} extend this approach to attention mechanisms and rate adaptation,
  and the broader semantic communication literature treats learned joint source-channel coding as a foundation for next-generation wireless systems~\cite{sagduyu2024will,getu2025semantic},
  where task-oriented, semantically-preserving transmission increasingly replaces strict bit-level fidelity~\cite{yang2022semantic,xie2021deep}.

  Such learned systems are being considered for mission-critical scenarios, where the decoder cannot be trusted blindly.
  Examples include bandwidth-constrained, safety-related image transmission for autonomous vehicles,
  industrial Internet of Things (IoT) sensing, and comparable applications~\cite{getu2025semantic,xie2021deep}.
  A corrupted reconstruction can propagate into a wrong downstream decision, \eg, an autonomous vehicle misjudging an obstacle from a degraded camera image, or an industrial controller acting on a misread sensor image~\cite{getu2025semantic,xie2020lite}.
  The worst-case behavior of a hand-crafted codec under channel noise is explicitly defined and can be systematically characterized.
  In contrast, the behavior of a DNN decoder is not explicitly bounded:
  the \djscc{} training objective does not constrain how much reconstruction quality may degrade for channel realizations outside the training set,
  but only optimizes average performance over observed ones.
  This gap presents a specific reliability concern, \eg,
  for a deployed \djscc{} receiver, it is important to determine the extent to which reconstruction quality may degrade under a specified channel disturbance.
  Evaluating \djscc{} models using finite-sample testing, which is the standard practice~\cite{bourtsoulatze2019deep,yang2023witt},
  can reveal individual failures
  but does not provide bounds on the worst-case performance across the full range of possible channel realizations.
  Similar to conventional software testing, this approach can demonstrate the presence of faults but cannot guarantee their absence.

  Neural network verification (NNV) can prove absence instead of merely searching for presence of these adversarial inputs.
  Bound-propagation verifiers~\cite{wang2018formal,zhang2018efficient,singh2019abstract,xu2020automatic} compute a certified interval on a network's own output,
  \eg, the logit range for a classifier, guaranteed to hold for every input in a specified region simultaneously, without enumerating individual inputs.
  State-of-the-art NNV tools now scale to networks with millions of parameters~\cite{katz2019marabou,duong2024harnessing,duong2025compositional}
  % and rank among the top performers at the annual neural-network verification competitions (VNN-COMPs)~\cite{kaulen2025vnncomp,brix2024fifth},
  with solvers such as \deeppoly{}~\cite{singh2019abstract}, \acrown{}~\cite{xu2020automatic,xu2021fast}, and interval bound propagation (\ibp{})~\cite{wang2018formal} covering
  various DNN architectures from image classifiers to aircraft collision-avoidance controllers~\cite{katz2017reluplex}.

  Recent research has extended formal certification methods to wireless systems.
  \citeauthor{kim2025certified}~\cite{kim2025certified} certified an antenna-selection network for massive multiple-input multiple-output (MIMO) systems,
  our previous works verified a mobile-traffic regression model~\cite{le2026traffic} and a power-control policy network~\cite{le2026power}.
  However, no formal verification framework currently exists for the \djscc{} decoder under physical wireless channel models.
  Despite these advancements, existing bound-propagation tools cannot verify the \djscc{} decoder in its current form.

  Three specific challenges hinder progress in this area.
  First, the \djscc{} decoder employs \prelu{}~\cite{he2015delving} activations throughout, allowing each layer to learn a small negative-side slope rather than fixing it at zero.
  Current verification tools support Rectified Linear Unit (\relu{}) activations but not \prelu{}. Replacing \prelu{} with \relu{} reduces reconstruction quality by up to two decibels (\autoref{sec:experiment:compatibility}).
  Second, the decoder uses transposed convolutions to upsample features, a function not natively supported by \acrown{}~\cite{xu2020automatic}. Naively replacing them alters the network.

  Last but not least, under flat Rayleigh fading assumption, each latent symbol acquires the same residual phase rotation remaining after amplitude equalization at the receiver, in addition to independent additive noise~\cite{tse2005fundamentals}.
  Standard interval encoding treats each latent coordinate as an independently.
  Consequently, interval's widths increase with phase uncertainty for each coordinate, despite a single shared rotation affecting all symbols simultaneously.
  Structured-perturbation verification~\cite{duong2026verifying,duong2026dual} demonstrates that bounding the generators of a perturbation, rather than its per-coordinate effects, produces significantly tighter certificates. This principle is adapted to the fading channel in the following section.

  In this paper, we introduce the first bound-propagation framework for the formal verification of \djscc{}'s decoder, addressing all three challenges while maintaining the reconstruction quality for which \djscc{} was originally trained.
  The primary contributions of this paper are:
  \begin{itemize}
    \item \textbf{\prelu{} relaxation.}
      We extend the optimization of linear relaxation \acrown{} to encompass the feasible relaxation slopes of \prelu{} (\autoref{sec:propose:prelu}).
      The standard \relu{} relaxation emerges as a special case when the learned negative-side slope is zero.
      Our derivation generalizes this scenario without loss and preserves the $1.78$ to $2.22$ dB of peak signal-to-noise ratio (PSNR) that would otherwise be lost by substituting \relu{}.
    \item \textbf{Upsample-then-convolution substitution.}
      We show that nearest-neighbor upsampling followed by convolution is a restricted transposed convolution (\autoref{prop:upsampleconv}),
      and the retrained decoder loses no more than $0.35$ dB of PSNR (\autoref{sec:experiment:compatibility}).
    \item \textbf{Structural Rayleigh input property.}
      We introduce the fading transform as a linear layer and directly bound its two real coordinates.
      This approach reduces the verifier's fading input from $2k$ per-symbol intervals
      to only $2$ degrees of freedom of a single complex block fading coefficient(\autoref{sec:propose:rayleigh}).
    \item \textbf{GloRo instantiation for \djscc{}.}
      We apply GloRo Lipschitz regularization~\cite{leino2021globally} to \djscc{}'s decoder,
      providing a closed-form radius under both additive white Gaussian noise (AWGN) and Rayleigh fading.
      This enables Lipschitz-aware fine-tuning from a pretrained checkpoint (\autoref{sec:propose:gloro}).
      This work aligns with efforts that train networks to facilitate verification~\cite{xu2024training},
  specialized to a channel-dependent radius.
  \end{itemize}

  Experiments on \djscc{} model~\cite{bourtsoulatze2019deep} show that \gloro{} training drops the \texttt{large} decoder's Lipschitz constant from $215.6$ to $25.5$,
  which makes it certifiable for the first time: the upper bound of mean absolute error (MAE) falls from above $1$ at every signal-to-noise ratio (SNR) to $0.08$.
  Structural encoding lowers the median certified bound by up to $41\%$ across model sizes, from $0.46$ to $0.27$ on \texttt{large}.
  At the widest phase uncertainty we test ($\delta{=}10^\circ$),
  the full pipeline certifies $192$ safe cases against $19$ for \gloro{} training with interval encoding alone,
  a $10\times$ gain.
  Over-the-air tests on two ADALM-PlutoSDR transceivers confirm the certificate on real hardware:
  across the ten CIFAR-10 classes the worst observed error is $0.082$, against a certified bound of $0.128$.

  % The remainder of this paper is organized as follows.
  % \autoref{sec:prelim} reviews \djscc{}, NNV, and bound propagation.
  % \autoref{sec:propose} presents our four contributions in detail.
  % \autoref{sec:experiment} evaluates them across model sizes, channel models, and real hardware.
  % \autoref{sec:related} surveys related verifiers and their applications.
  % \autoref{sec:conclusion} concludes.

\section{Preliminaries}
  \label{sec:prelim}

  \subsection{\djscc{} System Model}
    \label{sec:prelim:djscc}
    We consider image transmission over a bandwidth-limited wireless channel, as shown in \autoref{fig:system-model}.
    The transmitter sends an input image $x \in \mathbb{R}^n$ to the receiver by mapping $x$ to a complex-valued channel input symbol $s \in \mathbb{C}^k$.
    This mapping is the encoder $\mathcal{E}: \mathbb{R}^n \rightarrow \mathbb{C}^k$ in \autoref{fig:system-model}, following the \djscc{} design~\cite{bourtsoulatze2019deep}.

    The input size $n$ is typically much larger than the compressed representation size $k$; following the literature on \djscc{}, we call the ratio $k/n$ the \emph{bandwidth compression ratio}.
    The encoder normalizes $s$ to satisfy an average power constraint,
    which the communication protocol imposes to avoid interference with other transmissions:
    \begin{equation}
      \frac{1}{k} \mathbb{E} [s * s] \le P
    \end{equation}
    The encoder itself is a sequence of neural network layers:
    convolutional neural network (CNN) layers,
    \prelu{} activations,
    and a final normalization layer that enforces the power constraint.
    CNN layers extract image features and reduce them to a compact representation, the complex channel input, while \prelu{} activations let the encoder learn a non-linear mapping from the source image to the coded wireless signal space.
    \begin{figure}[t]
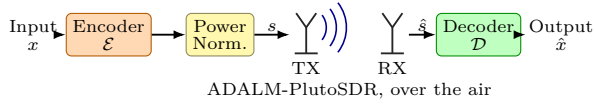

      \centering
      \drawSystemModel{0.95}
      \caption{The \djscc{} system model to be analyzed.}
      \label{fig:system-model}
    \end{figure}

    The wireless channel in \autoref{fig:system-model} carries $s$ to the receiver, corrupting the symbols with random noise.
    The receiver maps the corrupted channel output $\hat{s} \in \mathbb{C}^k$ back to an approximate reconstruction $\hat{x} \in \mathbb{R}^n$ of the original image;
    this mapping is the decoder $\mathcal{D}: \mathbb{C}^k \rightarrow \mathbb{R}^n$ in \autoref{fig:system-model}.
    The decoder mirrors the encoder.
    A stack of transposed convolutional layers with \prelu{} activations inverts the encoding.
    A final sigmoid layer then maps the output to $[0,1]$ before rescaling it to the pixel range.
    To recover a clean image, the decoder must also denoise the artifacts introduced by the channel.
    Together, the encoder, channel, and decoder form a joint source-channel coding model: the encoder and decoder are trained jointly, while the channel is a differentiable, non-trainable layer that connects them during training.
    Under AWGN, it applies $\hat{s} = s + n$ with $n \sim \mathcal{CN}(0, \sigma^2 I)$.

    Under slow Rayleigh fading, the channel instead applies $\hat{s} = hs + n$, where $h$ is the residual the receiver's channel estimator leaves rather than the raw fading coefficient.
    Equalizing by an estimate $\hat h$ that is accurate in amplitude leaves the unit-modulus rotation $h = e^{j\theta}$ with $|\theta| \le \delta$, where $\delta$ is the phase accuracy the estimator guarantees~\cite{tse2005fundamentals}.
    We assume amplitude equalization is exact, so $|h|=1$ and only the residual phase is unknown.

    \djscc{} evaluates quality by PSNR between $x$ and $\hat x$, computed from the mean squared error (MSE):
    \begin{align}
      \mathrm{PSNR}(x,\hat x) &= 10\log_{10}\!\left(\frac{x_{\max}^2}{\mathrm{MSE}(x,\hat x)}\right), \label{eq:psnr} \\
      \mathrm{MSE}(x,\hat x) &= \frac{1}{n}\sum_{i=1}^n (x_i-\hat x_i)^2, \label{eq:mse}
    \end{align}
    where $x_{\max}$ is the maximum possible pixel value.

  \subsection{Neural Network Verification}
    \label{sec:prelim:nnv}
    The channel corrupts the latent symbols before decoding, resulting in a lower-quality reconstruction of the original message (\eg, an image).
    Neural network verification can quantify how much this quality can degrade under a given noise range, providing a useful certificate,
    \eg, a proof that a minimum quality is always achieved.

    Formally, consider a DNN $N$ and a property $\phi$ of the form $\phi_{in} \rightarrow \phi_{out}$ over the inputs and outputs of $N$.
    The \emph{DNN verification problem} asks whether $\phi$ is valid for $N$.
    Bound-propagation methods~\cite{wang2018formal,singh2019abstract,xu2020automatic} solve this problem by over-approximating each layer.
    These methods are \emph{incomplete}:
    if the over-approximation still satisfies $\phi$, the property is proven;
    otherwise the verifier returns unknown,
    even when the true network may satisfy $\phi$.
    In contrast, \emph{complete} verifiers based on satisfiability modulo theories (SMT) solving~\cite{katz2017reluplex}, mixed-integer programming~\cite{tjeng2019evaluating}, branch-and-bound~\cite{wang2021beta}, or DPLL(t)~\cite{duong2025neuralsat} decide $\phi$ exactly: given enough time, they either prove the property or produce a genuine counterexample.
    Complete verifiers often use fast incomplete bound propagation as a subroutine to prune branches~\cite{xu2021fast, wang2021beta, duong2025neuralsat}.

  \subsection{Bound Propagation}
    \label{sec:prelim:bp}
    Bound propagation computes a sound interval $[\underline{y}, \overline{y}]$ on the output objective $y$ of $N$ over an input region $\phi_{in}$.
    In this paper, $\phi_{in}$ is an interval on the latent channel symbols.

    The main bound-propagation methods form a precision-cost trade-off.
    \ibp{}~\cite{wang2018formal} is fastest and loosest.
    \deeppoly{}~\cite{singh2019abstract,zhang2018efficient} back-substitutes linear relaxations in one pass.
    \acrown{}~\cite{xu2020automatic} further tunes the relaxation slopes by gradient descent, which is slower but yields the tightest bounds among incomplete verifiers.

    \subsubsection{Interval bound propagation}
      \ibp{}~\cite{wang2018formal} maintains only a concrete interval at each layer,
      starting from the input interval $\phi_{in} = [\underline{x}, \overline{x}]$.
      Consider an affine layer $z = Wx + b$ with incoming bounds $x \in [\underline{x}, \overline{x}]$.
      The pre-activation bounds are
      \begin{equation}
        \label{eq:ibp-linear}
        \begin{aligned}
          \underline{z} &= W^{+} \underline{x} + W^{-} \overline{x} + b, \\
          \overline{z} &= W^{+} \overline{x} + W^{-} \underline{x} + b,
        \end{aligned}
      \end{equation}
      where $W^{+} = \max(W, 0)$ and $W^{-} = \min(W, 0)$ are taken elementwise.
      The sign split is required: a negative weight reverses which endpoint of $[\underline{x}, \overline{x}]$ yields the minimum or maximum of that term.

      For $\hat{z} = \mathrm{ReLU}(z)$ with $z \in [\underline{z}, \overline{z}]$, monotonicity gives
      \begin{equation}
        \label{eq:ibp-relu}
        \underline{\hat{z}} = \mathrm{ReLU}(\underline{z}), \qquad \overline{\hat{z}} = \mathrm{ReLU}(\overline{z}).
      \end{equation}
      \ibp{} applies~\eqref{eq:ibp-linear}--\eqref{eq:ibp-relu} layer by layer to the output.
      Each step is fast but discards correlations among layer outputs.

      \autoref{fig:ibp-box} visualizes~\eqref{eq:ibp-relu} at one unstable \relu{}:
      the interval keeps only the endpoints $\underline{\hat{z}} = \mathrm{ReLU}(\underline{z}) = 0$ and $\overline{\hat{z}} = \mathrm{ReLU}(\overline{z}) = \overline{z}$,
      so the admissible $(z,\hat{z})$ pairs form the entire interval region
      $[\underline{z},\overline{z}] \times [\underline{\hat{z}},\overline{\hat{z}}]$ (shaded),
      which is the over-approximation area and discards the relation between $z$ and $\hat{z}$.
      The bounds therefore widen and often become vacuous on deep networks.

      \begin{figure}[t]
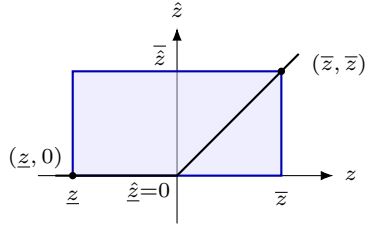

        \centering
        \drawIbpBox{1.15}
        \caption{Interval bound propagation through an unstable \relu{} over $[\underline{z},\overline{z}]$.}
        \label{fig:ibp-box}
      \end{figure}

    \subsubsection{Linear relaxation}
      Linear relaxation is tighter bound propagation methods by bounding each layer's pre-activation $z^{(l)}$ by affine functions of the network inputs $x$:
      \begin{equation}
        \underline{A}^{(l)} x + \underline{b}^{(l)} \;\le\; z^{(l)} \;\le\; \overline{A}^{(l)} x + \overline{b}^{(l)}.
        \label{eq:linear-bound}
      \end{equation}
      Affine and convolutional layers compose these bounds exactly.
      A \relu{} layer requires a local over-approximation.

      Consider a \relu{} layer with pre-activation $z \in [\underline{z}, \overline{z}]$,
      where $\underline{z} < 0 < \overline{z}$, and output $\hat{z} = \mathrm{ReLU}(z)$.
      The layer is \emph{stable} if $\underline{z} \ge 0$ (always active) or $\overline{z} \le 0$ (always inactive):
      then the map is exact and needs no relaxation.
      An unstable \relu{} is over-approximated by two linear functions (\autoref{fig:relu-wedge}).
      The tightest upper bound is the line through $(\underline{z}, 0)$ and $(\overline{z}, \overline{z})$:
      \begin{equation}
        \hat{z} \;\le\; \bar{d} z + \bar{b},
        \qquad
        \bar{d} = \tfrac{\overline{z}}{\overline{z}-\underline{z}},
        \qquad
        \bar{b} = -\tfrac{\underline{z} \overline{z}}{\overline{z}-\underline{z}}.
        \label{eq:relu-upper}
      \end{equation}
      The lower bound is any line $\hat{z} \ge \alpha z$ with $\alpha \in [0,1]$.

      In \autoref{fig:relu-wedge}, the solid blue line is the upper bound~\eqref{eq:relu-upper} and the thick dashed red line is the lower bound for one such $\alpha$;
      the arrow shows that $\alpha$ can rotate this line anywhere between the thin dashed lines $\alpha=0$ and $\alpha=1$.
      The shaded region is the over-approximation area for the drawn $\alpha$, and its shape depends on $\alpha$.
      \deeppoly{} picks an endpoint $\alpha = 0$ or $\alpha = 1$ to minimize the over-approximated area (shaded in blue).

      \begin{figure}[t]
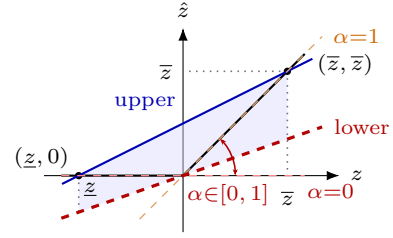

        \centering
        \drawReluWedge{1.15}
        \caption{Linear relaxation of an unstable \relu{} over $[\underline{z},\overline{z}]$ with lower slope $\alpha\in[0,1]$.}
        \label{fig:relu-wedge}
      \end{figure}

    \subsubsection{Back-substitution and concretization}
      Back-substitution is the procedure \deeppoly{} uses to eliminate intermediate layers: it repeatedly replaces each intermediate layer's variable by its lower or upper affine bound~\eqref{eq:linear-bound} until the resulting bound depends only on $x$.
      \emph{Concretization} then evaluates the linear forms over $\mathcal{B}$:
      \begin{equation}
        \label{eq:concretization}
        \begin{aligned}
          \underline{z}^{(l)} &= \min_{x \in \mathcal{B}}\;(\underline{A}^{(l)} x + \underline{b}^{(l)}), \\
          \overline{z}^{(l)} &= \max_{x \in \mathcal{B}}\;(\overline{A}^{(l)} x + \overline{b}^{(l)}).
        \end{aligned}
      \end{equation}
      For an $\ell_p$ ball $\mathcal{B} = \{x : \|x - x_0\|_p \le \varepsilon\}$, H\"{o}lder's inequality~\cite{wang2021beta} gives the closed form
      \begin{equation}
        \label{eq:holder-concretize}
        \begin{aligned}
          \overline{z}^{(l)} &= \overline{A}^{(l)} x_0 + \varepsilon \|\overline{A}^{(l)}\|_q + \overline{b}^{(l)}, \\
          \underline{z}^{(l)} &= \underline{A}^{(l)} x_0 - \varepsilon \|\underline{A}^{(l)}\|_q + \underline{b}^{(l)},
        \end{aligned}
      \end{equation}
      where $\tfrac{1}{p} + \tfrac{1}{q} = 1$.
      The resulting bound $[\underline{z}^{(l)}, \overline{z}^{(l)}]$ decide the neuron stability and the linear relaxation in the next layers.

    \subsubsection{Optimization of linear bound}
      \acrown{} treats the lower slopes $\alpha$ as trainable parameters in $[0,1]$.
      It optimizes them by gradient descent on the certified output bound.
      \autoref{alg:acrown} summarizes the procedure for an $L$-layer network that alternates an affine map and a \relu{} at each layer $l=1,\ldots,L$.
      \deeppoly{} is the special case with $T{=}0$ and a heuristic $\alpha\in\{0,1\}$.

      \begin{algorithm}[t]
        \small
        \caption{\acrown{} bound propagation}
        \label{alg:acrown}
        \Input{
          $L$-layer network with affine maps $\{W^{(l)},b^{(l)}\}_{l=1}^{L}$ and \relu{} layers;
          input ball $\mathcal{B}=\{x:\|x-x_0\|_p\le\varepsilon\}$;
          iterations $T$;
          step size $\eta$
        }
        \Output{output bounds $[\underline{z}^{(L)}, \overline{z}^{(L)}]$}
        Compute layer intervals $[\underline{z}^{(l)},\overline{z}^{(l)}]$ for $l=1,\ldots,L$ by \ibp{}~\eqref{eq:ibp-relu}\;
        % \tcp*{via~\eqref{eq:ibp-linear}--}
        Initialize slopes $\alpha^{(l)}\in[0,1]$ for each unstable \relu{}\;
        \For{$t\leftarrow 1$ \KwTo $T$}{
          \For{$l\leftarrow 1$ \KwTo $L$\label{line:acrown:outer}}{
            $(\underline{A},\underline{b}),(\overline{A},\overline{b})\leftarrow(I,0),(I,0)$\;
            \label{line:acrown:init}
            \For{$k\leftarrow l$ \KwTo $1$ by $-1$\label{line:acrown:back}}{
              $(\overline{d}^{(k)},\overline{b}^{(k)},\alpha^{(k)})\leftarrow\mathrm{Relax}\bigl(\underline{z}^{(k)},\overline{z}^{(k)},\alpha^{(k)}\bigr)$\;
              \label{line:acrown:relax}
              $(\underline{A},\underline{b}),(\overline{A},\overline{b})\leftarrow\mathrm{Combine}\bigl(\overline{d}^{(k)},\overline{b}^{(k)},\alpha^{(k)},(\underline{A},\underline{b}),(\overline{A},\overline{b})\bigr)$\;
              \label{line:acrown:subst}
              $(\underline{A},\underline{b})\leftarrow(\underline{A}\,W^{(k)}, \underline{A}\,b^{(k)}+\underline{b})$\;
              $(\overline{A},\overline{b})\leftarrow(\overline{A}\,W^{(k)}, \overline{A}\,b^{(k)}+\overline{b})$\;
              % \tcp*{affine map of layer $k$}
            }
            $\underline{z}^{(l)},\overline{z}^{(l)}\leftarrow\mathrm{Concretize}\bigl((\underline{A},\underline{b}),(\overline{A},\overline{b})\bigr)$~\eqref{eq:holder-concretize}\;
            \label{line:acrown:conc}
          }
          $\alpha\leftarrow\mathrm{Project}_{[0,1]}\bigl(\alpha+\eta \nabla_\alpha \underline{z}^{(L)}\bigr)$\;
          \label{line:acrown:opt}
          % \tcp*{ascent on the output lower bound; use $-\nabla_\alpha \overline{z}^{(L)}$ to tighten the upper bound}
        }
        \Return{$[\underline{z}^{(L)},\overline{z}^{(L)}]$}\;
      \end{algorithm}

      The algorithm first obtains one interval per layer with \ibp{}.
      Each iteration visits layers $l=1,\ldots,L$ (\autoref{line:acrown:outer}).
      For layer $l$, it back-substitutes from layer $k=l$ down to $1$ (\autoref{line:acrown:back}--\autoref{line:acrown:subst}).
      At each $k$, $\mathrm{Relax}$ turns the \relu{} bounds $[\underline{z}^{(k)},\overline{z}^{(k)}]$ and slope $\alpha^{(k)}$ into the upper line $(\overline{d}^{(k)},\overline{b}^{(k)})$ of \autoref{fig:relu-wedge}, and $\mathrm{Combine}$ folds $(\overline{d}^{(k)},\overline{b}^{(k)},\alpha^{(k)})$ into the running coefficients $(\underline{A},\underline{b}),(\overline{A},\overline{b})$ by the sign of each entry, before the affine map $(W^{(k)},b^{(k)})$ is applied.
      $\mathrm{Concretize}$~(\autoref{line:acrown:conc}) then evaluates the resulting linear form over $\mathcal{B}$ to yield the layer interval via~\eqref{eq:holder-concretize}.
      A projected gradient step~(\autoref{line:acrown:opt}) updates the slopes $\alpha$ to tighten the output bound.
      Optimizing $\alpha$ improves on fixed \deeppoly{} at the cost of $T$ full backward passes.

  \subsection{Lipschitz-Certified Training}
    \label{sec:prelim:gloro}
    Bound propagation formally evaluates a trained model: it certifies how much the output can change under a given input perturbation.
    It does not change the model weights, so it cannot by itself improve robustness.
    We therefore complement verification with \gloro{}-style training~\cite{leino2021globally}, which builds robustness into the network during learning.

    \begin{definition}[Lipschitz constant]
      \label{def:lipschitz}
      A function $N$ is Lipschitz continuous with constant $K$ (w.r.t.\ a chosen norm) if
      \begin{equation}
        \label{eq:lipschitz}
        \|N(x)-N(x')\| \;\le\; K \|x-x'\|
        \qquad\text{for all inputs $x,x'$.}
      \end{equation}
      The smallest such $K$ is the Lipschitz constant of $N$.
    \end{definition}
    Intuitively, $K$ is a global sensitivity: it upper-bounds how fast the output can change when the input moves.
    If $\|x-x'\|\le\varepsilon$, then~\eqref{eq:lipschitz} gives $\|N(x)-N(x')\|\le K\varepsilon$ \emph{everywhere}, not only near one test point.
    A small $K$ therefore quantifies \emph{global} robustness of the DNN: no pair of inputs within distance $\varepsilon$ can produce an output gap larger than $K\varepsilon$.

    \begin{definition}[Lipschitz constant of a DNN]
      \label{def:gloro-K}
      For an $L$-layer network, \gloro{} upper-bounds $K$ by the product of per-layer Lipschitz constants:
      \begin{equation}
        \label{eq:gloro-K}
        K \;\le\; \prod_{l=1}^{L} K^{(l)},
      \end{equation}
      where $K^{(l)}=\|W^{(l)}\|_2$ (the spectral norm) for an affine or convolutional layer, and $K^{(l)}$ is the activation's Lipschitz constant for a \relu{} or \prelu{} layer: $1$ for \relu{}, and $\max(1,|a|)$ for \prelu{} with negative slope $a$.
    \end{definition}

    Spectral norms are estimated during training by power iteration, which approximates $\|W^{(l)}\|_2$ by repeatedly multiplying a random unit vector by $W^{(l)\top}W^{(l)}$ until convergence.
    \gloro{} regularizes $K$ directly during training: it adds a penalty term on $K$ to the loss function,
    such that gradient descent shrinks the DNN's global sensitivity.
    \begin{equation}
      \label{eq:gloro-loss}
      \mathcal{L}_{\mathrm{GloRo}}(\theta) \;=\; \mathcal{L}(\theta) \;+\; \lambda\, g(K),
    \end{equation}
    for loss function $\mathcal{L}$, a regularization weight $\lambda$, and a monotonically increasing penalty $g$ on $K$.

\section{Problem Formulation}
  \label{sec:propose:problem}
  This section formulates \djscc{} verification as the DNN verification problem
  with an input property on the latent symbols and
  an output property on the reconstruction quality.

  \subsection{Input Property}
    \begin{definition}[Input property]
      \label{def:input-property}
      Let $s\in\mathbb{C}^k$ be the transmitted latent symbols
      and let $\hat{s}\in\mathbb{C}^k$ be the symbols that the decoder receives.
      Let $\|z\|_\infty := \max(|\mathrm{Re}(z)|,|\mathrm{Im}(z)|)$.
      The input property $\phi_{in}$ is an interval around each symbol of $\hat{s}$:
      \begin{equation}
        \label{eq:def:input-property}
        \phi_{in} \;=\; \bigl\{\hat s : \|\hat s_i - c_i\|_\infty \le r_i,\ i=1,\ldots,k\bigr\},
      \end{equation}
      for a center $c\in\mathbb{C}^k$ and radius $r\in\mathbb{R}_{\ge0}^k$ that depend on channel assumptions.
    \end{definition}

    \noindent\textbf{AWGN channel.}
    Under AWGN, $\hat s = s+n$ with $n\sim\mathcal{CN}(0,\sigma^2 I)$.
    The AWGN standard deviation on each axis of $n$ depends on the channel SNR $\gamma$ (in dB):
    \begin{equation}
      \label{eq:snr-sigma}
      \sigma(\gamma) = \sqrt{\tfrac{1}{2\cdot 10^{\gamma/10}}}.
    \end{equation}
    Let $\kappa$ be the number of standard deviations the interval extends on each axis, \eg, $\kappa=3$ covers about $99.7\%$ of the AWGN noise mass per axis, not jointly over all $2k$ axes.
    Finally, the parameter for input property are:
    \begin{equation}
      \label{eq:awgn-property}
      c = s, \quad r = \kappa\sigma(\gamma).
    \end{equation}

    \noindent\textbf{Rayleigh channel.}
    We assume slow fading, the regime where $h$ stays effectively constant over one packet's transmission.
    Under this assumption, $\hat{s} = h s + n$ where $h$ denotes the residual fading after imperfect channel estimation and equalization, \eg, $h=e^{j\theta}$, $|\theta|\le\delta$.
    The input domain for $h$ is the arc $\phi_h^\delta=\{e^{j\theta}:|\theta|\le\delta\}$.
    The reachable set $\{hs_i\}$ is the arc of radius $|s_i|=\sqrt{\mathrm{Re}(s_i)^2+\mathrm{Im}(s_i)^2}$ that $s_i$ sweeps as it rotates through $\pm\delta$.

    DNN verifiers take intervals as inputs, so we bound the reachable set $\{hs_i\}$ by an interval around $s_i$; \autoref{prop:rayleigh-sound} below shows this bound is sound, i.e., it never excludes a value $hs_i$ that some $h\in\phi_h^\delta$ can actually produce.

    \begin{proposition}[Residual fading interval is sound]
      \label{prop:rayleigh-sound}
      For $h=e^{j\theta}$ with $|\theta|\le\delta$ and $\delta\in[0,\pi]$,
      \[
        \|hs_i-s_i\|_\infty \;\le\; 2\sin\bigl(\tfrac{\delta}{2}\bigr)|s_i|.
      \]
    \end{proposition}
    \begin{proof}
      $|h-1|^2=(\cos\theta-1)^2+\sin^2\theta=2-2\cos\theta=4\sin^2(\tfrac{\theta}{2})$, so $|h-1|=2|\sin(\tfrac{\theta}{2})|$.
      Since $|\theta|\le\delta\le\pi$ and $\sin$ is increasing on $[0,\tfrac{\pi}{2}]$, $|h-1|\le2\sin(\tfrac{\delta}{2})$.
      Hence $|hs_i-s_i|=|h-1|\,|s_i|\le2\sin(\tfrac{\delta}{2})|s_i|$, and $\|hs_i-s_i\|_\infty\le|hs_i-s_i|$ since $\|z\|_\infty\le|z|$ for every $z\in\mathbb{C}$. \qedhere
    \end{proof}

    Centering the interval at $s_i$ and adding the $\kappa\sigma(\gamma)$ noise interval for $n_i$ gives the Rayleigh input property:
    \begin{equation}
      \label{eq:rayleigh-property}
      c_i = s_i, \quad r_i = 2\sin\bigl(\tfrac{\delta}{2}\bigr)|s_i| + \kappa\sigma(\gamma).
    \end{equation}
    %  This interval has the same width on both axes.
    % A tighter interval would use a different width on each axis, since a short arc is much wider along the tangent than along the radius; we keep the equal-width form here, and \autoref{sec:propose:rayleigh} improves on it.
    % Once $\kappa\delta\ge1$ (\eg, $\delta\gtrsim19.1^\circ$ at $\kappa=3$), the linear fading term alone already exceeds $|s_i|$, so we switch to the origin-centered interval $c_i=0$, $r_i=|s_i|+\kappa\sigma(\gamma)$ instead.
    % At $\delta=\pi$ this interval is exact for the full circle, so non-coherent reception is still covered correctly, if conservatively, over the whole $\kappa\delta\ge1$ range.
    Unlike the AWGN, the radius varies per transmitted symbol, and the fading term $2\sin(\tfrac{\delta}{2})|s_i|$ does not shrink with SNR.
    % Bounding each symbol independently also discards the fact that all $k$ symbols share one residual $h$ and \autoref{sec:propose:rayleigh} addresses this challenge.

  \subsection{Output Property}
    An output property need to that reflects reconstruction quality, yet bound-propagation verifiers only certify a scalar bound on a linear transformation of the network's output.
    \djscc{}'s natural quality metric, PSNR~\autoref{eq:psnr} and MSE~\autoref{eq:mse}, does not fit that requirement as they are nonlinear.
    We therefore bound the network's output by $\mathrm{MAE}$.

    \begin{definition}[Output property]
      \label{def:output-property}
      For a threshold $\tau>0$ and clean image $x\in\mathbb{R}^n$, the output property $\phi_{out}$ holds for a reconstruction $\hat x$ if
      \begin{equation}
        \label{eq:def:output-property}
        \mathrm{MAE}(x,\hat x) \;=\; \frac{1}{n}\sum_{i=1}^{n} |x_i-\hat x_i| \;\le\; \tau.
      \end{equation}
    \end{definition}

    For a scalar residual $r=x_i-\hat x_i$, the exact identity hold:
    \begin{equation}
      \label{eq:mae-relu-identity}
      |r| \;=\; \mathrm{ReLU}(r) + \mathrm{ReLU}(-r)
    \end{equation}

    Particularly, if $r\ge0$, $\mathrm{ReLU}(r)+\mathrm{ReLU}(-r)=r+0=|r|$; if $r<0$, $\mathrm{ReLU}(r)+\mathrm{ReLU}(-r)=0+(-r)=|r|$.
    $\mathrm{MAE}$ is therefore exactly representable with the \relu{} primitives our bound propagation already supports~\cite{duong2026verifying}.
    The property is \emph{safe} if no input in $\phi_{in}$ violates $\phi_{out}$.
    Bound propagation certifies an interval $[\underline{y},\overline{y}]$ on the output objective $y=\mathrm{MAE}(x,\hat x)$ over every $\hat{s} \in \phi_{in}$.
    If the over-approximation satisfies $\phi_{out}$, \eg, $\overline{y}\le\tau$, the result is safe.

  \subsection{Verification Framework}
    \begin{figure}[t]
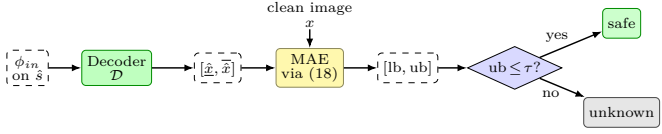

      \centering
      \resizebox{\linewidth}{!}{\drawVerificationFramework{1.0}}
      \caption{The verification framework for \djscc{}.}
      \label{fig:verification-framework}
    \end{figure}

    \autoref{fig:verification-framework} shows the verification framework for the receiver of \djscc{}.
    Bound propagation carries the input property $\phi_{in}$ through the decoder $\mathcal{D}$ to a certified interval on the reconstruction denotes $[\underline{\hat x},\overline{\hat x}]$.
    It then compares that interval against the clean image $x$ via the two-\relu{} $\mathrm{MAE}$ identity~\eqref{eq:mae-relu-identity}.
    The result is a certified interval $[\underline{y},\overline{y}]$ on $y=\mathrm{MAE}(x,\hat x)$.
    The final check against the output property $\phi_{out}$ (\autoref{def:output-property}) follows \autoref{sec:prelim:nnv}: $\overline{y}\le\tau$ returns safe, otherwise unknown.

\section{Propose Approach}
  \label{sec:propose}
  We make \djscc{} verifiable and certifiably robust with four contributions:
  a sound \prelu{} relaxation that extends \acrown{} to \djscc{}'s activations (\autoref{sec:propose:prelu}),
  a transposed-convolution replacement that \acrown{} supports natively (\autoref{sec:propose:convtranspose}),
  a structural Rayleigh input property that exposes $h$ to bound propagation (\autoref{sec:propose:rayleigh}),
  and a \gloro{} instantiation with a closed-form radius for both channels (\autoref{sec:propose:gloro}).

  \subsection{\prelu{} Relaxation with an Auto-Tuned Slope}
    \label{sec:propose:prelu}
    \djscc{} uses \prelu{} activations for every layer except the final sigmoid output layer,
    to learn a non-linear mapping from the source signal space to the coded signal space~\cite{bourtsoulatze2019deep}.
    \prelu{} generalizes \relu{} by learning the negative-side slope $a$ instead of fixing it at $0$.

    \begin{definition}[Parametric \relu{}]
      \label{def:prelu}
      For a pre-activation $z \in \mathbb{R}$ and a layer's learned negative-side slope $a$,
      \begin{equation}
        \label{eq:def:prelu}
        \mathrm{PReLU}(z) =
        \begin{cases}
          z, & z \ge 0, \\
          a z, & z < 0.
        \end{cases}
      \end{equation}
    \end{definition}
    Every \prelu{} layer in the encoder and decoder shares one such $a$.
    This single extra parameter per layer improves accuracy over \relu{}
        at negligible computational cost~\cite{he2015delving}.

    \deeppoly{}'s \relu{} relaxation (\autoref{sec:prelim:bp}) assumes a lower-bound feasible interval of $[0,1]$:
    any slope $\alpha\in[0,1]$, including the endpoint $\alpha=0$,
    gives a sound lower bound $\hat z\ge\alpha z$ on \relu{}'s output.
    Once the layer's learned slope $a>0$,
    reusing \relu{}'s linear relaxation is \emph{unsound} because
    if $z<0$ the true \prelu{} output is $\hat z = az < 0$.
    This, \relu{}'s relaxation cannot simply be reused as-is.

    \begin{figure}[t]
      \centering
      \drawPreluWedge{1.15}
      \caption{Linear relaxation of an unstable \prelu{} over $[\underline z,\overline z]$ with learned slope $a$.}
      \label{fig:prelu-wedge}
    \end{figure}

    \prelu{} is piecewise-linear but non-linear at the origin,
    so we over-approximate an unstable \prelu{} neuron with two linear functions
    over the pre-activation interval $z \in [\underline z,\overline z]$.
    This keeps back-substitution exact linear algebra.
    \autoref{fig:prelu-wedge} shows the resulting wedge, the \prelu{} counterpart of \autoref{fig:relu-wedge}.
    The upper bound is the line through $(\underline z, a\underline z)$ and $(\overline z, \overline z)$.
    The lower bound is again a line through the origin but its feasible slopes stay within $[a, 1]$.
    In \autoref{fig:prelu-wedge}, the solid blue line is the upper bound~\eqref{eq:prelu-bounds},
    the thick dashed red line is the lower bound for one slope in the feasible interval $[a,1]$,
    the arrow shows that this slope can rotate anywhere between the thin dashed lines of slope $a$ and slope $1$,
    and the shaded region is the over-approximation area for the drawn slope.

    \begin{definition}[\prelu{} relaxation]
      \label{def:prelu-relax}
      Let $z \in [\underline z,\overline z]$ and $\hat z = \mathrm{PReLU}(z)$. If $\underline z\ge0$ or $\overline z\le0$, the
      neuron is \emph{stable} and $\hat z=z$ or $\hat z=az$ exactly, with no
      relaxation. Otherwise ($\underline z<0<\overline z$), for every $\alpha\in[0,1]$,
      \begin{equation}
        \label{eq:prelu-bounds}
        \begin{aligned}
          &\bigl(a+\alpha(1-a)\bigr) z \;\le\; \hat z \;\le\; \bar d  z + \bar b, \\
          &\bar d = \tfrac{\overline z-a\underline z}{\overline z-\underline z}, \qquad \bar b = a\underline z - \bar d \underline z.
        \end{aligned}
      \end{equation}
    \end{definition}

    Setting $a=0$ in \autoref{eq:prelu-bounds} gives exactly \relu{}'s line
    $\hat z\le\bar d z+\bar b$~\eqref{eq:relu-upper} and lower family $\hat{z} \ge \alpha z$,
    $\alpha\in[0,1]$: a nonzero $a$ only rescales that same $\alpha$ onto $[a, 1]$
    instead of $[0,1]$.
    So \autoref{eq:prelu-bounds} is a drop-in replacement for
    $\mathrm{Relax}$ at \autoref{line:acrown:relax} of \autoref{alg:acrown} wherever
    that algorithm reaches a \prelu{} layer.
    \acrown{} still treats $\alpha$ as trainable, projected onto $[a, 1]$.
    A bound is \emph{sound} if it is guaranteed to contain the true value: a sound interval $[\underline y,\overline y]$ satisfies $\underline y\le y\le\overline y$ for every value $y$ the network can actually produce, even if $[\underline y,\overline y]$ is wider than necessary.
    Soundness is what makes a certificate trustworthy~\cite{duong2025generating}: an unsound bound could silently miss a bad case, the way an underestimated link margin could silently miss an outage.
    \relu{}'s relaxation is sound over $\alpha\in[0,1]$; \autoref{prop:prelu-sound} establishes the guarantee for \prelu{}.

    \begin{proposition}[\prelu{} relaxation soundness]
      \label{prop:prelu-sound}
      For $\underline z<0<\overline z$, $a\le1$, $z\in[\underline z,\overline z]$, and
      $\alpha\in[0,1]$ we have
      \[
        \bigl(a+\alpha(1-a)\bigr)z \;\le\; \mathrm{PReLU}(z) \;\le\; \bar d z+\bar b,
      \]
      where $\bar d=\tfrac{\overline z-a\underline z}{\overline z-\underline z}$ and
      $\bar b=a\underline z-\bar d \underline z$.
    \end{proposition}
    \begin{proof}
      Write $d=a+\alpha(1-a)\in[a,1]$. If $z<0$, then $dz\le az=\mathrm{PReLU}(z)$ since
      $d\ge a$. If $z\ge0$, then $dz\le z=\mathrm{PReLU}(z)$ since $d\le1$.
      If $z<0$, then $\mathrm{PReLU}(z)=az\le\bar d z+\bar b$ because
      $az\le\bar d z+\bar b \Leftrightarrow \overline z(a-1)(z-\underline z)\le0$, which
      holds since $\overline z>0$, $a\le1$, and $z\ge\underline z$.
      If $z\ge0$, then $\mathrm{PReLU}(z)=z\le\bar d z+\bar b$ because
      $z\le\bar d z+\bar b \Leftrightarrow (1-a)\underline z(\overline z-z)\le0$, which
      holds since $\underline z<0$, $a\le1$, and $z\le\overline z$. \qedhere
    \end{proof}

    The endpoints $\alpha\in\{0,1\}$ also cover the stable cases of \autoref{def:prelu-relax} ($d=a$ and $d=1$).
    Setting $a=0$ recovers \relu{}'s line and feasible interval $[0,1]$, so \autoref{prop:prelu-sound} strictly generalizes \deeppoly{}'s \relu{} soundness.

  \subsection{Substituting Transposed Convolution with Upsampling}
    \label{sec:propose:convtranspose}
    The decoder in~\cite{bourtsoulatze2019deep} decode the latent symbols with transposed convolutional layers,
    each layer doubling the spatial resolution until the output matches the source image.
    Transposed convolution's output-to-input map is transposed and fractionally strided,
    a sparsity pattern ordinary convolution does not share, so it needs its own linear-relaxation rule.

    Even ordinary convolution never forms this matrix explicitly,
    since it would not fit in a commercial graphics processing unit (GPU)'s memory;
    \acrown{}~\cite{xu2020automatic} propagates bounds through the convolution's own patch representation.
    A transposed-convolution relaxation would need an analogous patch-mode rule on top of the relaxation itself,
    which is a difficult engineering task.
    We substitute every decoder transposed convolution with nearest-neighbor upsampling followed by an ordinary convolution,
    both operations our bound propagation handles exactly or near-exactly.
    \autoref{prop:upsampleconv} below shows the substitution stays within the same operator family,
    stated in one spatial dimension with one channel; the two-dimensional multi-channel case applies the same identity per index and channel.
    Let $(u \ast v)[n] = \sum_{j} u[j] v[n-j]$ denote convolution.

    \begin{definition}[Transposed convolution]
      \label{def:convtranspose}
      The stride-$s$ transposed convolution of an input $x$ with kernel $w$ is
      \begin{equation}
        \label{eq:def:convtranspose}
        y[n] \;=\; \sum_{m} x[m]  w[n - sm] .
      \end{equation}
      Each input entry $x[m]$ places a copy of the kernel at position $sm$, so the output resolution grows by a factor of $s$.
    \end{definition}

    \begin{definition}[Upsample-then-convolution]
      \label{def:upsampleconv}
      Nearest-neighbor upsampling by $s$ copies each entry $s$ times: $x_{\mathrm{up}}[j] = x[\lfloor j/s \rfloor]$.
      Convolution with kernel $v$ then follows:
      \begin{equation}
        \label{eq:def:upsampleconv}
        y \;=\; x_{\mathrm{up}} \ast v .
      \end{equation}
    \end{definition}

    \begin{proposition}[Upsample-then-convolution is a transposed convolution]
      \label{prop:upsampleconv}
      For every kernel $v$, nearest-neighbor upsampling by $s$ followed by convolution with $v$ equals the stride-$s$ transposed convolution with kernel $w = \mathbf{1}_s \ast v$, where $\mathbf{1}_s$ is the all-ones kernel of length $s$.
    \end{proposition}
    \begin{proof}
      Let $x_0$ insert $s-1$ zeros between entries: $x_0[j] = x[j/s]$ if $s$ divides $j$ and $x_0[j]=0$ otherwise.
      First, $(x_0 \ast w)[n] = \sum_j x_0[j]  w[n-j] = \sum_m x[m]  w[n-sm]$, which is~\eqref{eq:def:convtranspose}: convolving the zero-inserted input is exactly the stride-$s$ transposed convolution.
      Second, $x_{\mathrm{up}} = x_0 \ast \mathbf{1}_s$: entry $j$ of $x_0 \ast \mathbf{1}_s$ sums $x_0[j-i]$ over $0 \le i < s$, and the only multiple of $s$ among $j-s+1,\dots,j$ is $s\lfloor j/s\rfloor$, leaving $x[\lfloor j/s\rfloor]$.
      By associativity, $x_{\mathrm{up}} \ast v = (x_0 \ast \mathbf{1}_s) \ast v = x_0 \ast (\mathbf{1}_s \ast v)$, the stride-$s$ transposed convolution with kernel $\mathbf{1}_s \ast v$. \qedhere
    \end{proof}

    By \autoref{prop:upsampleconv}, the substituted layer produces the same output shape
    and stays in the same operator family as the transposed convolution it replaces.
    We empirically confirm this substitution preserves \djscc{}'s reconstruction performance (\autoref{sec:experiment}).

  \subsection{Structural Perturbation of Fading Residual Interval}
    \label{sec:propose:rayleigh}
    The Rayleigh interval~\eqref{eq:rayleigh-property} is sound but loose:
    each per-symbol radius $r_i=2\sin(\tfrac{\delta}{2})|s_i|+\kappa\sigma(\gamma)$ is derived from $h$'s worst case for that symbol alone,
    though all $k$ symbols share one residual $h$.
    The true reachable set is therefore a one-dimensional curve inside $\mathbb{R}^{2k}$,
    yet the per-symbol interval inflates it to a full $2k$-dimensional hypercube.

    Recent work on structured-perturbation verification~\cite{duong2026verifying,duong2026dual} argues
    that whenever a perturbation is generated by fewer degrees of freedom than the space it inhabits,
    the verifier should reason over the generators directly.
    Prepending the generative map as a DNN layer and bounding only the low-dimensional input
    then lets bound propagation track the structural constraint exactly, without inter-coordinate slack.
    \autoref{sec:prelim:djscc} already states the exact channel formula $\hat s=hs+n$, while \autoref{eq:rayleigh-property} only certifies a fixed, per-symbol over-approximation of its effect.
    We propose replacing $\hat s$ as the verifier's input with the two quantities that generate it, $h$ and $n$: $h$ is bounded by its support $\phi_h^\delta$ (\autoref{def:input-property}), and $n$ keeps its existing noise-interval bound.
    We prepend a layer to the network graph that computes $hs$ directly, broadcasting the single $h$ across all $k$ symbols, and recover the original formula $\hat s=hs+n$ by adding the existing noise term $n$.

    $h\mapsto hs$ is exactly linear in $h$, since $s$ is the fixed, known latent for the instance being verified.
    In real coordinates it is an exact block-diagonal $2k\times2$ matrix multiply, one rotation-scaling block $M_i$ per symbol, all sharing the same $h$:
    \begin{equation}
      \label{eq:rotation-scaling-block}
      \begin{pmatrix} \mathrm{Re}(hs_i) \\ \mathrm{Im}(hs_i) \end{pmatrix}
      = \begin{pmatrix} \mathrm{Re}(s_i) & -\mathrm{Im}(s_i) \\ \mathrm{Im}(s_i) & \mathrm{Re}(s_i) \end{pmatrix}
      \begin{pmatrix} \mathrm{Re}(h) \\ \mathrm{Im}(h) \end{pmatrix}.
    \end{equation}
    This affine layer is exactly what our bound propagation algorithms already handles.
    The addition with $n$ is exact as it already is for the AWGN case.

    The fading arc $\phi_h^\delta=\{e^{j\theta}:|\theta|\le\delta\}$ (solid line in \autoref{fig:fading-arc}) is enclosed by independent intervals
    $\mathrm{Re}(h)\in[\cos\delta,1]$, $\mathrm{Im}(h)\in[-\sin\delta,\sin\delta]$ (shaded).
    \begin{figure}[t]
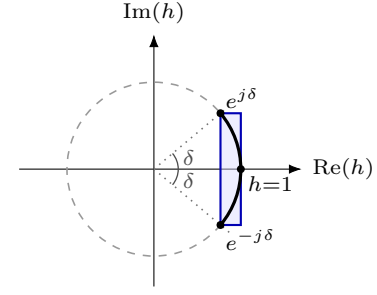

      \centering
      \drawFadingArc{1.15}
      \caption{Fading residual $h$ and its interval enclosure.}
      \label{fig:fading-arc}
    \end{figure}
    Exposing $h$ turns the joint input into the product domain
    \begin{equation}
      \label{eq:rayleigh-joint-property}
      \phi_{in} \;=\; \phi_h^\delta \times \bigl\{n : \|n\|_\infty \le \kappa\sigma(\gamma)\bigr\}.
    \end{equation}
    The joint input is then enclosed by an interval on $(h,n)$,
    so Hölder concretization~\eqref{eq:holder-concretize} ($p{=}\infty$, $q{=}1$) applies.

  \subsection{GloRo for DJSCC}
    \label{sec:propose:gloro}
    We instantiate \gloro{} for \djscc{}'s decoder $\mathcal{D}$: the task loss $\mathcal{L}$ is the MSE reconstruction loss, and the penalty $g(K_\mathcal{D})$ is derived below from the input property.

    The decoder is built entirely from convolutions, some preceded by upsampling (\autoref{sec:propose:convtranspose}), \prelu{} activations, and a final sigmoid.
    \autoref{eq:gloro-K}'s product $K_\mathcal{D}\le\prod_{l=1}^{L_\mathcal{D}}K_\mathcal{D}^{(l)}$ extends to the whole decoder with two more per-layer constants.
    Nearest-neighbor upsampling by $2$ in each spatial dimension replicates every pixel into a $2\times2$ block, so its Lipschitz constant is $2$.
    The output sigmoid has Lipschitz constant $\tfrac14$, its maximum slope.
    Every layer's constant is known in closed form, so $K_\mathcal{D}$ is cheap to compute.
    We derive the radius for \gloro{} training from the input property:
    by the Lipschitz definition \autoref{def:lipschitz},
    $\|\mathcal{D}(\hat s)-\mathcal{D}(s)\|_2\le K_\mathcal{D}\|\hat s-s\|_2\le K_\mathcal{D} r$,
    where $r$ is the $\ell_2$ radius of the input property around the clean latent $s$:
    the smallest value with $\|\hat s-s\|_2\le r$ for every $\hat s\in\phi_{in}$ (\autoref{def:input-property}).
    Under AWGN, each of the $2k$ real axes of $\hat s-s=n$ stays within $\kappa\sigma(\gamma)$ (\autoref{eq:awgn-property}), so $r_{\mathrm{AWGN}}=\sqrt{2k}\,\kappa\sigma(\gamma)$.

    Because $\mathrm{MSE}$ averages over the $n$ entries of $x$, the triangle inequality gives, for every $\hat s\in\phi_{in}$,
    $\sqrt{\mathrm{MSE}(x,\mathcal{D}(\hat s))}\le\sqrt{\mathrm{MSE}(x,\mathcal{D}(s))}+K_\mathcal{D} r/\sqrt{n}$.
    Only the second term depends on $K_\mathcal{D}$, so we add its square to the reconstruction loss:
    \begin{equation}
      \label{eq:gloro-djscc-loss}
      \mathcal{L}_{\mathrm{GloRo}}(\theta) \;=\; \mathrm{MSE}(x,\mathcal{D}(\hat s)) \;+\; \lambda \frac{(K_\mathcal{D} r)^2}{2k},
    \end{equation}
    where $\hat s$ is sampled through the channel as in regular training, and the penalty equals $(K_\mathcal{D} r)^2/n$ up to the constant factor $n/(2k)$, which $\lambda$ absorbs.
    We fine-tune from a pretrained checkpoint, updating only the decoder.
    A smaller $K_\mathcal{D}$ tightens the certified $[\underline{y},\overline{y}]$ that DNN verifier computes.

    Under Rayleigh fading, we keep the same loss form~\eqref{eq:gloro-djscc-loss}.
    By \autoref{prop:rayleigh-sound}, $|h-1|\le2\sin(\tfrac{\delta}{2})$, and the $hs$-layer linearity (\autoref{sec:propose:rayleigh}) gives $\|\hat s-s\|_2\le2\sin(\tfrac{\delta}{2})\|s\|_2+\sqrt{2k}\,\kappa\sigma(\gamma)$.
    With $\|s\|_2=\sqrt{k}$ (\autoref{sec:prelim:djscc}), the Rayleigh radius is:
    \begin{equation}
      \label{eq:gloro-rayleigh-radius}
      r_{\mathrm{Rayleigh}}(\gamma,\delta) \;=\; \sqrt{k}\,\bigl(2\sin(\tfrac{\delta}{2})+\sqrt{2}\,\kappa\sigma(\gamma)\bigr).
    \end{equation}
    Substituting $r_{\mathrm{Rayleigh}}$ for $r$ in \eqref{eq:gloro-djscc-loss} adapts \gloro{} training to Rayleigh fading.

\section{Experiments}
  \label{sec:experiment}

  Our evaluation\footnote{Code: \url{https://anonymous.4open.science/r/verify-djscc-33C8/}} answers the following research questions:
  \noindent\textbf{RQ1}: Do our verifier-compatibility choices preserve \djscc{}'s reconstruction quality? \\
  \noindent\textbf{RQ2}: How does \djscc{}'s certified robustness vary across different channel assumptions? \\
  \noindent\textbf{RQ3}: How much does \gloro{} tighten the certified bound? \\
  \noindent\textbf{RQ4}: Does structural encoding of property help? \\
  \noindent\textbf{RQ5}: Does theoretical safe certificate transfer to over-the-air software-defined-radio (SDR) transmission? \\

  \subsection{Experimental Design}
    \label{sec:experiment:setup}

    \noindent\textbf{Models and data.}
    We evaluate three decoder sizes (\autoref{tab:model-sizes}):
    \texttt{small}, \texttt{medium}, and \texttt{large},
    all with the upsample-then-convolution decoder of \autoref{sec:propose:convtranspose}.
    Models are trained on CIFAR-10~\cite{krizhevsky2009learning} for $100$ epochs under a randomized training-time SNR curriculum ($\gamma\sim\mathcal{U}(1,13)$ dB).

    We formally verify the models at SNRs $\gamma\in\{10,15,20,25,30\}$ dB.
    Under Rayleigh fading, the decoder additionally trains under a randomized residual-phase curriculum ($\delta\sim\mathcal{U}(0^\circ,30^\circ)$ per batch,
    \autoref{sec:prelim:djscc}); verification and \gloro{} certification both sweep estimator accuracy at $\delta\in\{0^\circ,2.5^\circ,5^\circ,7.5^\circ,10^\circ\}$.

    \begin{table}[t]
      \centering
      \caption{Decoder configurations evaluated.}
      \label{tab:model-sizes}
      \begin{tabular}{@{}lccc@{}}
        \toprule
        & \texttt{small} & \texttt{medium} & \texttt{large} \\
        \midrule
        Latent dimension $k$ & 32 & 192 & 256 \\
        Encoder layers & 3 & 4 & 5 \\
        Decoder layers & 4 & 4 & 5 \\
        Parameters & 7.7k & 29.1k & 143.4k \\
        \bottomrule
      \end{tabular}
    \end{table}

    \noindent\textbf{Verifiers.}
    We compare \ibp{}, \deeppoly{}, and \acrown{}, all implemented in our own bound-propagation Python library
    since existing libraries do not cover the \prelu{} relaxation of \autoref{sec:propose:prelu} out of the box.
    On standard fully connected and convolutional \relu{} networks, our implementations produce bounds that match auto\_LiRPA~\cite{xu2020automatic}.

    \noindent\textbf{Training methods.}
    \naive{} training minimizes $\mathrm{MSE}$ loss with Adam~\cite{kingma2014adam}.
    \gloro{} training follows \autoref{sec:propose:gloro}~\cite{leino2021globally}, fine-tuning a \naive{}-trained checkpoint for 20 epochs.

    Spectral norms use 10 power-iteration steps per weight matrix per training step.
    We set $\lambda{=}10^{-4}$ in \autoref{eq:gloro-djscc-loss}: higher values collapse PSNR across all three model sizes, while lower values leave $K_\mathcal{D}$ unchanged; $\lambda{=}10^{-4}$ reduces $K_\mathcal{D}$ while preserving reconstruction quality for all sizes.

    \noindent\textbf{Metrics.}
    We use PSNR to compare training runs.
    Certified bounds are on MAE (\autoref{def:output-property}):
    every $[\underline{y},\overline{y}]$ below bounds $\mathrm{MAE}$, not $\mathrm{PSNR}$,
    and we report the upper bound $\overline{y}$.

    \noindent\textbf{Runtime.}
    All experiments run on a single NVIDIA RTX 4080 SUPER GPU (16 GB), paired with an AMD Ryzen 9 5900X 12-core central processing unit (CPU) and 128 GB of DDR4-3200 random-access memory (RAM).

    \noindent\textbf{SDR hardware.}
    The over-the-air experiment (\autoref{sec:experiment:sdr}) uses two ADALM-PlutoSDR transceivers on a 915 MHz carrier at a 2 MHz sample rate.
    A Raspberry Pi 4 controls both transceivers, one transmitting and one receiving, and stores the received frames for offline decoding.
    We attenuate the transmitter by 20 dB and hold the receiver gain fixed at 30 dB.
    Higher settings saturated the receiver's analog-to-digital converter, and the resulting clipping destroyed frame coherence.

    Each frame carries one OFDM symbol, a 512-point inverse fast Fourier transform (IFFT) with a 64-sample cyclic prefix, behind an 11-chip Barker preamble for frame synchronization.
    Every fourth subcarrier carries a known binary phase-shift keying (BPSK) pilot, which leaves 384 free bins, of which the first $k{=}256$ carry the \djscc{} latents.
    The receiver estimates the channel from those pilots by least squares, interpolates it across bins, and equalizes amplitude and phase together.
    Before the image transmissions, 300 probe frames of complex Gaussian tones characterize the link.
    We wipe the phase from each frame, then take the 99th percentile of the per-tone residual phase as $\delta$ and the error power as the operating SNR.

  \subsection{Effect of Verifier-Compatibility Changes}
    \label{sec:experiment:compatibility}

    We isolate each verifier-compatibility change from \autoref{sec:propose} with a paired ablation against the original \djscc{} decoder (\prelu{}, transposed convolution), training all three model sizes under the AWGN channel with \naive{} training (\autoref{sec:experiment:setup}).
    \autoref{fig:rq1} reports test PSNR and MAE for the original decoder,
    the upsample-then-convolution substitution (\autoref{sec:propose:convtranspose}),
    and a \relu{} decoder in place of \prelu{} (\autoref{sec:propose:prelu}).
    Replacing \prelu{} with \relu{} cost $1.78$, $2.01$, and $2.22$ dB of PSNR on \texttt{small}, \texttt{medium}, and \texttt{large} respectively,
    a $23\%$, $29\%$, and $28\%$ increase in MAE.
    The penalty grew with model size because a deeper decoder stacks more activation layers, and each one discards the negative-side signal that its learned slope would otherwise carry.
    This supports \autoref{sec:propose:prelu}'s claim that \prelu{}'s single extra parameter per layer is worth its added relaxation cost.
    Verifying a \relu{}-substituted decoder would answer a question about a different, worse model, which is why we derive a \prelu{} relaxation instead of substituting the activation away.
    The upsample-then-convolution substitution left reconstruction quality essentially intact: PSNR moved by at most $0.35$ dB across all sizes ($+0.05$ dB on \texttt{small}, $-0.35$ dB on \texttt{medium}, $-0.28$ dB on \texttt{large}), and MAE by under $7\%$ at every size.
    The sign of the change varied with model size, which suggests the residual gap reflects training variation rather than a systematic loss.

    This is consistent with \autoref{prop:upsampleconv}: both layers belong to the same operator family.
    \autoref{sec:propose:convtranspose}'s claim therefore holds: the substitution buys a verifiable decoder at a small cost.
    \begin{figure}[t]
      \centering
      \begin{subfigure}[b]{0.48\linewidth}
        \centering
        \includegraphics[width=\linewidth]{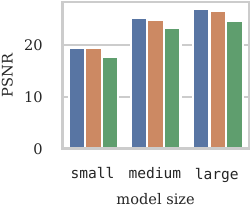}
        \caption{$\mathrm{PSNR}$}
        \label{fig:rq1-psnr}
      \end{subfigure}
      \hfill
      \begin{subfigure}[b]{0.48\linewidth}
        \centering
        \includegraphics[width=\linewidth]{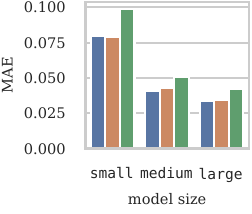}
        \caption{$\mathrm{MAE}$}
        \label{fig:rq1-mae}
      \end{subfigure}
      \\[0.2em]
      \includegraphics[width=\linewidth]{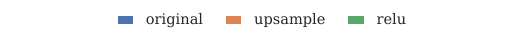}
      \caption{Reconstruction quality of the original, upsample-then-convolution, and \relu{} decoders.}
      \label{fig:rq1}
    \end{figure}

  \subsection{Certified Robustness Across Channel Models}
    \label{sec:experiment:channel}
    \begin{figure}[t]
      \centering
      \includegraphics[width=\linewidth]{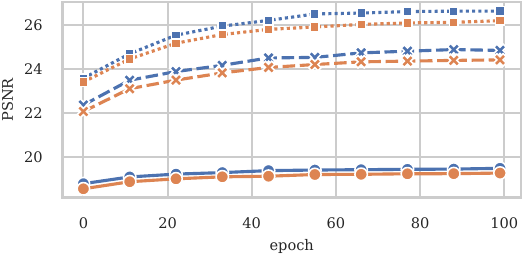}
      \\[0.2em]
      \includegraphics[width=\linewidth]{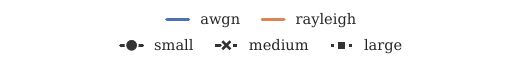}
      \caption{Test PSNR over 100 training epochs for all three model sizes under AWGN and Rayleigh fading.}
      \label{fig:rq2-convergence}
    \end{figure}
    \autoref{fig:rq2-convergence} reports test PSNR at a fixed evaluation SNR of 10 dB, the midpoint of the training curriculum $\gamma \sim \mathcal{U}(1, 13)$ dB.
    All models improved rapidly during the first 20 epochs, then slowed and plateaued.
    The \texttt{small} model achieved 99\% of its peak PSNR by epochs 35 to 38, while the \texttt{medium} and \texttt{large} models converged between epochs 60 and 65.

    Rayleigh-trained models demonstrated similar initial improvement rates but plateaued at lower performance levels.
    This outcome is attributed to the fading residual $h = e^{j\theta}$ (\autoref{sec:prelim:djscc}), which represents a phase rotation absent in AWGN training.
    Consequently, the Rayleigh channel is harder to decode.
    For the \texttt{small}, \texttt{medium}, and \texttt{large} models, Rayleigh-trained variants exhibited mean squared error (MSE) increases of $5.1\%$, $11.4\%$, and $10.5\%$, respectively, compared to their AWGN counterparts, corresponding to PSNR reductions of $0.22$, $0.47$, and $0.43$ dB.

    Model size influenced final PSNR more than channel type.
    The \texttt{small} model, using $94.6\%$ fewer parameters than \texttt{large}, achieved a PSNR that was $6.9$ dB lower under both channel conditions.
    The \texttt{medium} model, with $79.7\%$ fewer parameters ($4.9\times$ smaller), achieved PSNR values within $1.6$ dB of \texttt{large} under AWGN and within $1.8$ dB under Rayleigh fading.
    \begin{figure}[t]
      \centering
      \includegraphics[width=0.8\linewidth]{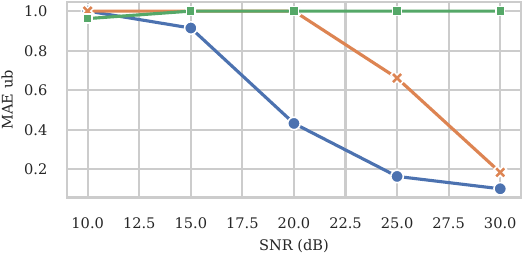}
      \\[0.2em]
      \includegraphics[width=0.8\linewidth]{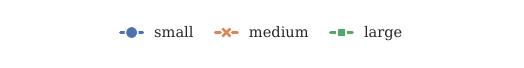}
      \caption{Certified MAE upper bound $\overline{y}$ under AWGN vs.\ SNR, per model size (\deeppoly{}).}
      \label{fig:rq2-awgn}
    \end{figure}

    \begin{figure}[t]
      \centering
      \includegraphics[width=\linewidth]{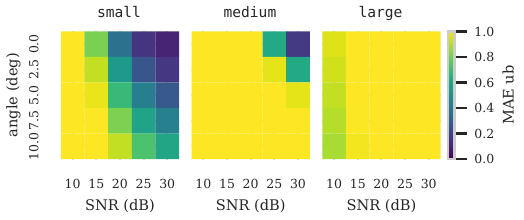}
      \caption{Certified MAE upper bound $\overline{y}$ under Rayleigh fading vs.\ SNR and residual phase $\delta$, per model size (\deeppoly{}).}
      \label{fig:rq2-rayleigh}
    \end{figure}

    \autoref{fig:rq2-awgn} shows the certified MAE upper bound $\overline{y}$ under AWGN.
    For \texttt{small} ($k{=}32$), $\overline{y}$ fell from $1.23$ at SNR $10$ dB to $0.10$ at $30$ dB.
    For \texttt{medium} ($k{=}192$), $\overline{y}$ stayed above $1$ until SNR $25$ dB and reached $0.19$ at $30$ dB.
    For \texttt{large} ($k{=}256$), $\overline{y}$ never dropped below $1$: bound propagation across $256$ latent dimensions accumulated too much interval error for the verifier to certify anything useful.
    Certification is more tractable for smaller models, as fewer latent dimensions constrain interval expansion.
    \autoref{fig:rq2-rayleigh} extends the analysis to Rayleigh fading, sweeping SNR and residual phase $\delta \in \{0^\circ, 2.5^\circ, 5^\circ, 7.5^\circ, 10^\circ\}$.
    At $\delta{=}0^\circ$, the bound matched the AWGN result; as $\delta$ increases, the fading term $2\sin(\delta/2)|s_i|$ in the input radius (\autoref{eq:rayleigh-property}) widens the input property and loosens the bound.

    For \texttt{small} at $30$ dB, $\overline{y}$ rose from $0.10$ at $\delta{=}0^\circ$ to $0.59$ at $\delta{=}10^\circ$.
    For \texttt{medium}, the bound was only useful at $\delta{=}0^\circ$ ($\overline{y}{=}0.17$ at $30$ dB); any residual phase pushed $\overline{y}$ above $0.6$.
    For \texttt{large}, $\overline{y}$ exceeded $1$ in all conditions.
    SNR and $\delta$ jointly determine $\overline{y}$.
    Higher SNR reduces the noise radius, whereas larger $\delta$ increases the fading radius, resulting in competing effects.

    The dominant factor between SNR and $\delta$ has significant implications for system design.
    For the \texttt{small} model at the highest tested SNR, increasing the residual phase error from a perfect estimate to $10^\circ$ raised $\overline{y}$ from $0.10$ to $0.59$.
    This represents an approximate sixfold increase that cannot be mitigated by increasing transmit power, as the noise radius is negligible at $30$ dB and the fading radius remains independent of SNR.
    Therefore, when a link operates well above the noise floor, the channel estimator's accuracy, rather than transmit power, determines the strength of the certificate the receiver can provide.
    This observation also clarifies why per-symbol interval encoding becomes the limiting constraint at large $\delta$, thereby motivating the structural encoding approach evaluated in \autoref{sec:experiment:structural}.

  \subsection{Effect of \gloro{} Training on Certified Robustness}
    \label{sec:experiment:gloro}
    \autoref{fig:rq3-convergence} presents \gloro{} fine-tuning initiated from the naive checkpoint across 20 epochs.
    The figure tracks test PSNR, the Lipschitz constant $K_\mathcal{D}$, and robust loss for each model size under AWGN and Rayleigh fading.
    The Lipschitz constant $K_\mathcal{D}$ decreased substantially during the first five epochs, from $41.9$ to $6.5$ for \texttt{small} and from $215.6$ to $25.5$ for \texttt{large}.
    All three model sizes stabilized by epoch 10.
    Twenty epochs of fine-tuning were sufficient, as the remaining ten epochs resulted in minimal changes to $K_\mathcal{D}$.

    The PSNR cost associated with regularizing $K_\mathcal{D}$ was $0.72$ dB for \texttt{small}, $1.14$ dB for \texttt{medium}, and $2.01$ dB for \texttt{large} under AWGN.
    The corresponding Rayleigh fading costs were similar, at $0.73$, $1.09$, and $2.11$ dB, respectively.
    Larger models incurred greater PSNR costs because their $K_\mathcal{D}$ values were higher and required more contraction.
    \begin{figure}[t]
      \centering
      \includegraphics[width=\linewidth]{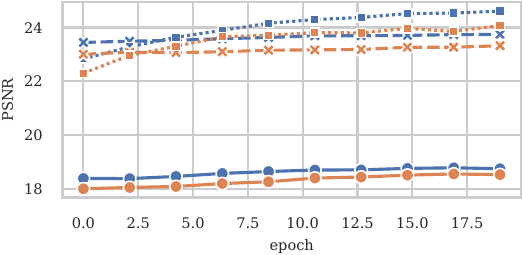}
      \\[0.2em]
      \includegraphics[width=\linewidth]{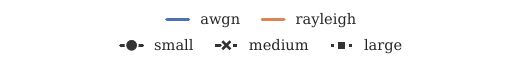}
      \caption{\gloro{} fine-tuning over 20 epochs per model size under AWGN and Rayleigh fading.}
      \label{fig:rq3-convergence}
    \end{figure}

    \autoref{fig:rq3-cdf} plots the certified MAE upper bound $\overline{y}$ as a cumulative distribution function (CDF) per model size, comparing naive and \gloro{} training under AWGN and Rayleigh fading ($\delta{=}0^\circ$, all SNRs pooled).
    For \texttt{small} under AWGN, $\overline{y}$ fell from $1.23$ to $0.53$ at SNR $10$ dB, a $57\%$ reduction.
    At SNR $30$ dB, naive and \gloro{} both reached $\overline{y} \approx 0.10$: the AWGN noise ball dominates the bound at high SNR, so tightening $K_\mathcal{D}$ adds nothing there.
    For \texttt{medium}, the gain concentrated at high SNR: $\overline{y}$ dropped from $0.19$ to $0.08$ at $30$ dB.
    For \texttt{large}, naive training never certified anything useful ($\overline{y} > 1$ at every SNR), whereas \gloro{} brought $\overline{y}$ to $0.85$ at $10$ dB and $0.08$ at $30$ dB, making \texttt{large} certifiable in our experiments.
    At $\delta{=}0^\circ$, Rayleigh bounds fell by comparable fractions.
    \gloro{} shifted the entire bound distribution to the left for every model size and both channels, indicating that the improvement is not limited to a subset of images.

    The two effects of \gloro{} oppose each other, and the dominant effect depends on the specific model.
    For \texttt{large}, a PSNR cost of $2.01$ dB resulted in the difference between no certificate and obtaining a certificate at every tested SNR, making this trade-off advantageous.
    For \texttt{small} at $30$ dB, naive and \gloro{} approaches achieved $\overline{y}\approx0.10$, indicating that the $0.72$ dB cost did not yield a tighter bound.

    This outcome occurs because $\overline{y}$ is determined by the term that dominates the input radius.
    When the noise ball is dominant, as at high SNR on a small latent, reducing $K_\mathcal{D}$ cannot further tighten a bound already constrained by the input region.
    \gloro{} is most effective when the bound's looseness stems primarily from the decoder's amplification rather than the channel, which characterizes the operational regime of all larger models.
    \begin{figure}[t]
      \centering
      \begin{subfigure}[b]{0.32\linewidth}
        \centering
        \includegraphics[width=\linewidth]{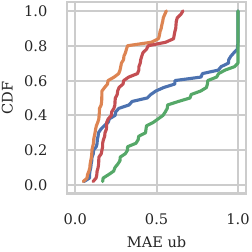}
        \caption{\texttt{Small}}
        \label{fig:rq3-cdf-small}
      \end{subfigure}\hfill
      \begin{subfigure}[b]{0.32\linewidth}
        \centering
        \includegraphics[width=\linewidth]{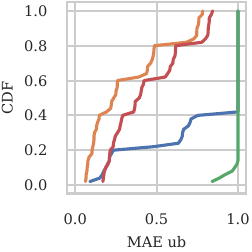}
        \caption{\texttt{Medium}}
        \label{fig:rq3-cdf-medium}
      \end{subfigure}\hfill
      \begin{subfigure}[b]{0.32\linewidth}
        \centering
        \includegraphics[width=\linewidth]{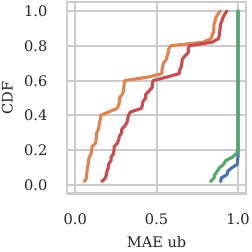}
        \caption{\texttt{Large}}
        \label{fig:rq3-cdf-large}
      \end{subfigure}
      \\[0.2em]
      \includegraphics[width=\linewidth]{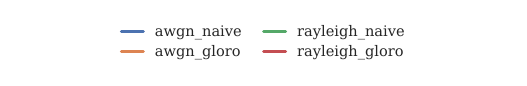}
      \caption{CDF of the certified MAE upper bound $\overline{y}$ for naive and \gloro{}-trained decoders.}
      \label{fig:rq3-cdf}
    \end{figure}

  \subsection{Joint Verification of the Fading Coefficient}
    \label{sec:experiment:structural}

    \autoref{fig:rq4-angle-line} compares interval and structural encoding as the residual phase $\delta$ grows from $0^\circ$ to $10^\circ$ on the \texttt{small} model at SNR $20$ dB.
    At $\delta{=}0^\circ$, both approaches coincided at $\overline{y}\approx 0.32$: there is no phase uncertainty, so the per-symbol interval is already tight.
    As $\delta$ increases, the interval bound rises sharply, reaching $\overline{y}\approx 0.52$ at $\delta{=}10^\circ$.
    This occurs because each of the $2k$ interval half-widths increases by $2\sin(\delta/2)|s_i|$ (\autoref{eq:rayleigh-property}), and these independent excesses accumulate.
    The structural bound stays nearly flat over the same range, at $\overline{y}\approx 0.33$ when $\delta{=}10^\circ$.
    Reasoning over the single shared $h$ (\autoref{sec:propose:rayleigh}) provides the verifier with two additional degrees of freedom, independent of $k$.
    As a result, the phase term contributes minimal slack.
    \autoref{fig:rq4-boxplot} extends the comparison across all three model sizes (S/M/L $=$ \texttt{small}/\texttt{medium}/\texttt{large}), pooling all SNRs and phases $\delta\in\{0^\circ,2.5^\circ,5^\circ,7.5^\circ,10^\circ\}$.
    For \texttt{small}, the structural median fell from $0.28$ to $0.19$; for \texttt{medium}, from $0.43$ to $0.27$; for \texttt{large}, from $0.46$ to $0.27$.

    The improvement increased with model size.
    Larger $k$ results in more per-symbol interval terms, so the fixed-dimensional input of structural encoding saves more accumulated slack.
    The narrowing of interquartile ranges indicates that the benefit is consistent across SNR and phase conditions, rather than influenced by a small number of instances.
    \begin{figure}[t]
      \centering
      \begin{subfigure}[b]{0.48\linewidth}
        \centering
        \includegraphics[width=\linewidth]{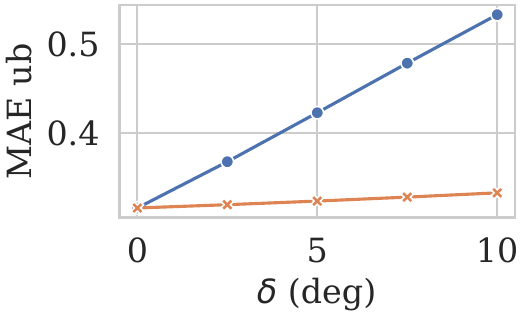}
        \\[0.2em]
        \includegraphics[width=\linewidth]{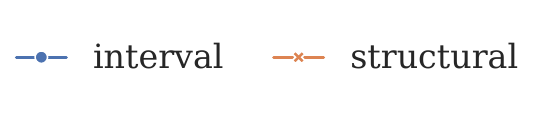}
        \caption{$\overline{y}$ vs.\ $\delta$ (\texttt{small}, SNR $20$ dB).}
        \label{fig:rq4-angle-line}
      \end{subfigure}
      \hfill
      \begin{subfigure}[b]{0.48\linewidth}
        \centering
        \includegraphics[width=\linewidth]{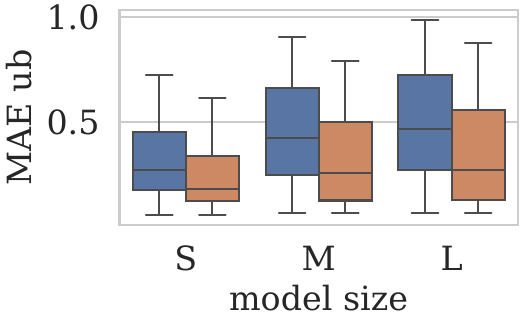}
        \\[0.2em]
        \includegraphics[width=\linewidth]{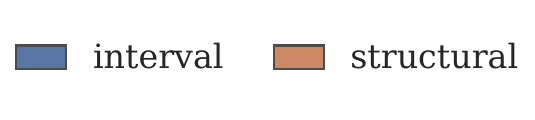}
        \caption{$\overline{y}$ per model size (all SNRs and $\delta$ pooled).}
        \label{fig:rq4-boxplot}
      \end{subfigure}
      \caption{Interval vs.\ structural encoding of the fading coefficient under Rayleigh fading.}
      \label{fig:rq4}
    \end{figure}

    Structural encoding substantially reduces $\delta$-sensitivity and lowers the bound across all model sizes (\autoref{fig:rq4}).
    Structural encoding replaces $2k$ interval inputs with a single affine layer over two real variables (\autoref{eq:rotation-scaling-block}).
    Consequently, the verifier propagates a smaller input region through one additional exact layer and performs less computation overall.
    This advantage contrasts with verification scenarios, where larger networks present greater computational challenges.
    The removed penalty scales with $k$, so structural encoding provides the greatest benefit for models where interval encoding is least effective.
    Together with \autoref{sec:experiment:channel}, these findings demonstrate that the fading radius limits the certificate once a link surpasses the noise floor.
    Structural encoding prevents this radius from scaling with the latent dimension.

  \subsection{Certified Safety Count of the Full Pipeline}
    \label{sec:experiment:cactus}
    The preceding sections measured $\overline{y}$ in isolation.
    A practical deployment question requires a broader perspective: specifically, the number of (image, threshold) pairs a verifier can certify and the time required to do so.
    We counted certified-safe pairs on the \texttt{large} model at $\delta{=}10^\circ$, the widest phase uncertainty we consider.

    Ten images, one from each CIFAR-10 class, combined with five SNRs, resulted in $50$ bound queries.
    Each query was evaluated against seven MAE thresholds $\tau\in\{0.25,0.275,\ldots,0.40\}$, resulting in $50{\times}7$ cactus instances per encoding.
    A cactus instance is a (property, network, threshold) query.
    \acrown{} is applied to each $\tau$ individually and terminates at the first pass where $\overline{y}<\tau$.

    \autoref{fig:rq5-cactus} presents a comparison between \deeppoly{} and \acrown{} on the \gloro{}-trained model, pooling results from both interval and structural encodings.
    Both verifiers certified the same $211$ instances: $19$ on interval and $192$ on structural encoding.
    \deeppoly{} completed each query in approximately $0.10$\,s.
    \acrown{} returned in about $0.11$\,s, because the first DeepPoly-equivalent pass met every $\tau$ that these runs cleared.
    Therefore, \deeppoly{} serves as the preferred default when the property is specified at these thresholds.
    \ibp{} did not certify any instance in any of the training or encoding configurations evaluated.
    Computation speed is inconsequential if the relaxation method fails to certify any instances.

    \autoref{fig:rq5-cactus-2} compares interval and structural encodings under \acrown{}, aggregating results from both \naive{} and \gloro{} training.
    Naive training did not certify any instances under either encoding.
    In the absence of Lipschitz regularization, the decoder amplified the input region excessively, preventing $\overline{y}$ from falling below the specified thresholds.
    \gloro{} combined with interval encoding certified $19$ instances.
    \gloro{} with structural encoding certified $192$ instances, $10\times$ more.
    Neither component substitutes for the other.
    \gloro{} is essential; without it, no instances are certified due to the dominant amplification effect of the decoder.
    Structural encoding converts the available headroom into certified instances at large $\delta$ by eliminating the per-symbol expansion that would re-inflate the input region.

    This ordering is consistent with the results presented in \autoref{sec:experiment:gloro}.
    \gloro{} addresses the decoder's contribution to the bound, while structural encoding targets the channel's contribution.
    These two methods operate on distinct terms, resulting in complementary rather than overlapping improvements.
    \begin{figure}[t]
      \centering
      \begin{subfigure}[b]{0.48\linewidth}
        \centering
        \includegraphics[width=\linewidth]{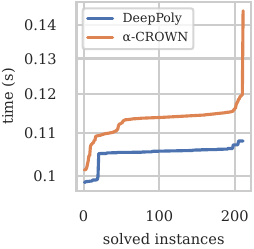}
        \caption{\deeppoly{} vs.\ \acrown{}.}
        \label{fig:rq5-cactus}
      \end{subfigure}
      \hfill
      \begin{subfigure}[b]{0.48\linewidth}
        \centering
        \includegraphics[width=\linewidth]{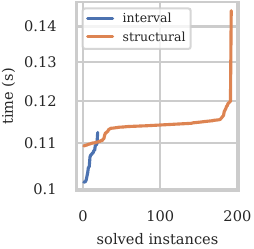}
        \caption{Interval vs.\ structural.}
        \label{fig:rq5-cactus-2}
      \end{subfigure}
      \caption{Certified-safe instance count vs.\ verification time.}
      \label{fig:rq5}
    \end{figure}

  \subsection{Over-the-Air Transfer to SDR}
    \label{sec:experiment:sdr}

    All previously established bounds are based on a channel model. This section evaluates whether a certificate derived from that model remains valid when applied to a real radio link.
    We characterize the channel between the two ADALM-PlutoSDR transceivers (\autoref{sec:experiment:setup}), verify the decoder at the measured operating point, then transmit and check the reconstructions against the bound.
    Pilot measurements yielded an operating SNR of $27.4$ dB and a $99$th-percentile residual phase of $\delta{=}5.4^\circ$, both within the ranges explored in \autoref{sec:experiment:channel}.
    Although the link exhibited minimal interference, its SNR significantly exceeded the $1$--$13$ dB range used during training, indicating that the decoder operated outside its original training conditions.
    \autoref{fig:rq6-grid} presents a representative input image and its corresponding over-the-air reconstruction at this operating point for each of the ten CIFAR-10 classes.
    \begin{figure}[t]
      \centering
      \includegraphics[width=\linewidth]{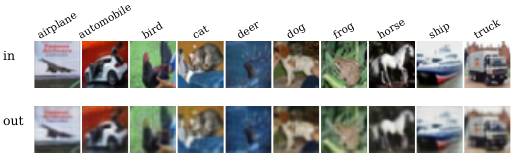}
      \caption{Input images (top) and over-the-air reconstructions (bottom) for all ten CIFAR-10 classes.}
      \label{fig:rq6-grid}
    \end{figure}

    \autoref{fig:rq6-stack} displays, for each class, the observed over-the-air reconstruction error (the maximum $\mathrm{MAE}$ over $30$ transmissions of a representative image) alongside its margin to the certified upper bound $\overline{y}$, such that each complete bar represents $\overline{y}$.
    The certificate held for every class: the observed maximum stayed below $\overline{y}$ in each class and in the aggregate, where the worst of $300$ frames ($0.0824$) stayed below the largest bound ($0.1278$).
    Class $7$ approached its bound most closely ($0.082$ versus $0.128$), while class $5$ exhibited the widest margin ($0.035$ versus $0.091$).
    Across all classes, the observed maximum ranged from $38$ to $65\%$ of the certified bound.
    The certificate is neither violated nor excessively permissive; it provides meaningful headroom without being so loose as to permit arbitrary behavior.

  Each bar aggregates $30$ transmissions of a representative image, for a total of $300$ over-the-air frames shown in the figure. Within each class, the $99$th percentile and the maximum differed by no more than $0.0001$.
    Reconstruction error demonstrated high repeatability: across $30$ independent transmissions through the radio link, the worst trial was only marginally greater than the median.
    The margin in \autoref{fig:rq6-stack} is stable slack against the channel model rather than the tail of a noisy measurement.
    This experimental design introduces one caveat.
    Each bar corresponds to a fixed image per class, so it evaluates the certificate at ten specific points rather than across the entire test set.
    The bound applies to every input within the noise region; the measurement confirms its validity at the sampled points.
    \begin{figure}[t]
      \centering
      \includegraphics[width=\linewidth]{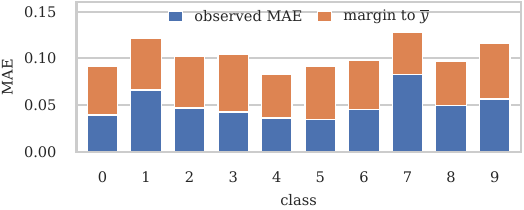}
      \caption{Per-class observed over-the-air $\mathrm{MAE}$ stacked with its margin to the certified bound $\overline{y}$.}
      \label{fig:rq6-stack}
    \end{figure}

    This guarantee is inherently conditional.
    It constrains the reconstruction error only while the link remains within the SNR and phase envelope measured by the pilot signals.
    A receiver operating outside this envelope requires a new certificate.
    This requirement is not onerous in practice, as re-verification at a new operating point can be performed in well under a second using \deeppoly{}, as described in \autoref{sec:experiment:cactus}.
    Consequently, a receiver can re-certify itself whenever its channel estimate changes, incurring a computational cost significantly lower than that of the protected transmission.
    The measurement demonstrates that the entire process—from channel modeling, through bound-propagation certification, to the physical radio link—remains valid end to end.
    The bounds the verifier produces offline continue to constrain the hardware's behavior accurately.

\section{Related Work}
  \label{sec:related}

  \subsection{Neural Network Verifiers}
    \label{sec:related:verifiers}
    Complete verifiers perform an exact search of the network's piecewise-linear structure. Reluplex~\cite{katz2017reluplex} extends the Simplex algorithm to \relu{} constraints. Marabou~\cite{katz2019marabou} incorporates theory solvers and parallel divide-and-conquer strategies. The DPLL(T)-based NeuralSAT~\cite{duong2023dpllt,duong2025neuralsat} explores activation patterns using bound propagation as its theory solver.
    Other approaches improve verification scalability by exploiting neuron stability~\cite{duong2024harnessing}, training networks to promote stable neurons~\cite{xu2024training}, decomposing networks into assume-guarantee subproblems~\cite{duong2025compositional}, and learning branching heuristics through reinforcement learning~\cite{duong2026reinforcement}. Our \gloro{} instantiation (\autoref{sec:propose:gloro}) adopts the training-based approach.

    Bound propagation, the foundation of this work, over-approximates each layer using a sound linear relaxation and back-substitutes these relaxations to the input.
    \ibp{}~\cite{wang2018formal} offers the fastest but least precise bounds. \deeppoly{}~\cite{singh2019abstract} back-substitutes symbolic linear bounds in a single pass. \crown{}~\cite{zhang2018efficient} and its successors \acrown{}~\cite{xu2020automatic,xu2021fast} and $\beta$-\crown{}~\cite{wang2021beta} optimize relaxation slopes using gradient descent.
    However, none of these tools provide native support for \prelu{} or transposed convolution (\autoref{sec:propose:prelu}, \autoref{sec:propose:convtranspose}).
    Bound propagation tools have won every recent International Verification of Neural Networks Competition (VNN-COMP)~\cite{kaulen2025vnncomp}. This consistent success motivates our decision to build upon this family of methods.

  \subsection{Applications of DNN Verification}
    \label{sec:related:applications}
    A growing line of work certifies wireless and telecommunication systems.
    Kim et al.~\cite{kim2025certified} certify robustness of an antenna-selecting network for massive MIMO; Le and Matsumura~\cite{le2026traffic} verify a regression network predicting mobile-traffic load; Le et al.~\cite{le2026power} verify a power-control policy network for massive MIMO; Le et al.~\cite{le2025fggm} use formal methods to craft adversarial attacks on a deep reinforcement learning (DRL)-based multi-user MIMO scheduler.
    Each targets a discrete decision or a single continuous quantity under one noise source.
    Our setting differs in three ways.
    We certify \djscc{}'s decoder under the physical AWGN and Rayleigh fading channel models.
    The property is a continuous reconstruction bound.
    Finally, no prior wireless verification addresses \prelu{}, transposed convolutions, or the per-symbol interval blow-up of a shared fading coefficient, which block off-the-shelf verifiers (\autoref{sec:intro}).

\section{Conclusion}
  \label{sec:conclusion}
  We addressed three primary challenges that prevent off-the-shelf verifiers from operating with \djscc{}'s decoder.
  \prelu{}, transposed convolution, and a fading channel shared across all latent symbols each allow for a principled reformulation compatible with existing bound-propagation methods: a relaxation that extends the \relu{} case, a layer substitution within the same operator family, and an adjustment of verifier input to the channel's actual low-dimensional generator.
  None of these reformulations sacrifices reconstruction quality for verifiability; in particular, structural encoding strengthens certificates. The resulting certificate remains valid under real radio channel conditions.

  Experimental results support these findings: \gloro{} training reduces the \texttt{large} decoder's Lipschitz constant from $215.6$ to $25.5$, enabling certification - the upper bound decreases from above $1$ at all SNRs to $0.08$.
  Over-the-air tests further confirm it: across the ten CIFAR-10 classes the worst observed error is $0.082$, against a certified bound of $0.128$.

  Applying this approach to attention-based and rate-adaptive successors, such as WITT~\cite{yang2023witt} and SwinJSCC~\cite{yang2024swinjscc}, as well as to modalities beyond images, remains an open research direction.
  These successors are based on vision transformers (ViTs), whose attention layers require abstract domains beyond the \relu{} and \prelu{} relaxations used in this work.
  Recent abstract domains developed for verifying large generative pre-trained transformer (GPT) models~\cite{duong2026zonogpt} provide a promising foundation for this extension.

\printbibliography

\begin{IEEEbiography}[{\includegraphics[width=1in,height=1.25in,clip,keepaspectratio]{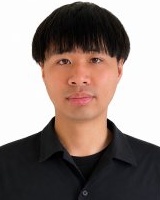}}]{THANH LE}{\space}received the B.E. degree in information and communication technology and the M.S. degree in data science from Hanoi University of Science and Technology, Vietnam, in 2019 and 2021, respectively, and the Ph.D. degree in informatics from The Graduate University for Advanced Studies, SOKENDAI, Tokyo, Japan, in 2025.
Since 2025, he has been a Researcher with the Wireless Systems Laboratory, National Institute of Information and Communications Technology (NICT), Japan.
His research focuses on bringing formal methods from software engineering to analyze, verify, and improve the robustness of AI-enabled wireless systems.
\end{IEEEbiography}

\begin{IEEEbiography}[{\includegraphics[width=1in,height=1.25in,clip,keepaspectratio]{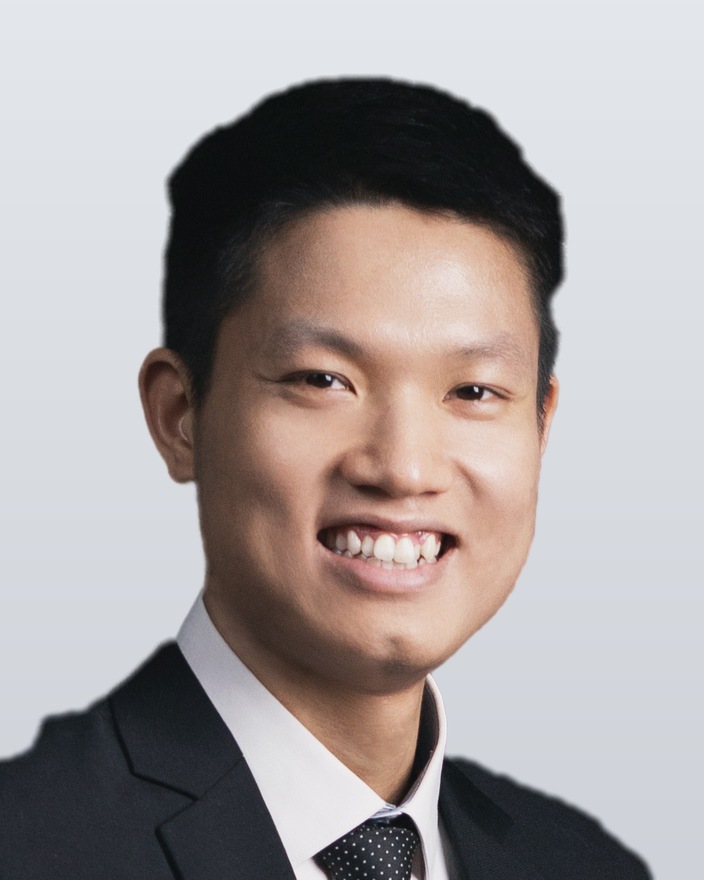}}]{HAI DUONG}{\space}received the B.S. and M.S. degrees in electrical engineering from Hanoi University of Science and Technology, Hanoi, Vietnam, in 2019 and 2021, respectively.
He is currently pursuing the Ph.D. degree in computer science with George Mason University, Fairfax, VA, USA, where he works at the ROARS Lab on deep neural network verification.
He develops NeuralSAT, a neural network verification tool that ranked second overall in VNN-COMP 2024 and 2025.
His research interests include DNN verification and testing, proof generation and checking, formal methods, and automated reasoning.
\end{IEEEbiography}

\begin{IEEEbiography}[{\includegraphics[width=1in,height=1.25in,clip,keepaspectratio]{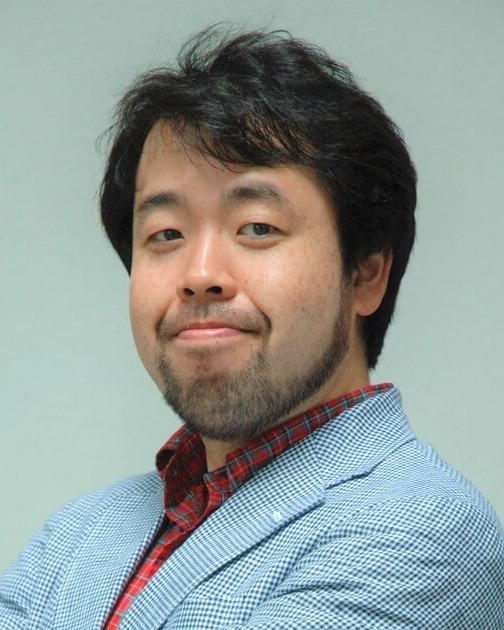}}]{TAKESHI MATSUMURA}{\space}received the M.S. and Ph.D. degrees in Electronic Engineering and Nano-mechanics Engineering from Tohoku University in 1998 and 2020, respectively.
From 1998 to 2007, he was engaged in the R\&D of wireless communications devices in several private companies.
He joined National Institute of Information and Communications Technology (NICT) in 2007, focusing on white-space and 5G systems.
From 2016 to 2019, he has served as Associated Professor at Graduate School of Informatics, Kyoto University.
He returned to NICT in April to take on the role of Research Manager and became Director of Wireless Systems Laboratory in 2021.
He now serves as Research Executive Director at Network Research Institute and Project Manager of the Moonshot Goal 1 project.
His research interests include white-space communication, wide-area wireless networks, 5G systems, and reliability-ensuring platform for cybernetic avatar teleoperation.
He is a member of IEICE and IEEE.
\end{IEEEbiography}

\begin{IEEEbiography}[{\includegraphics[width=1in,height=1.25in,clip,keepaspectratio]{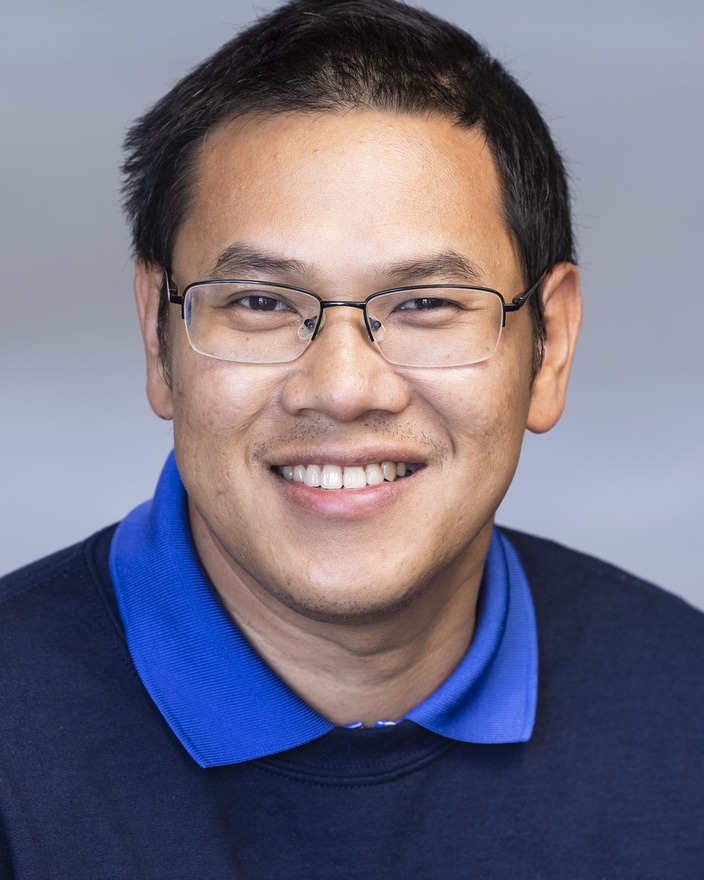}}]{THANHVU NGUYEN}{\space}received the B.S. and M.S. degrees in computer science from The Pennsylvania State University, USA, in 2003 and 2006, respectively, and the Ph.D. degree in computer science from the University of New Mexico, Albuquerque, USA, in 2014.
He was a Postdoctoral Researcher with the University of Maryland, College Park, from 2014 to 2016, and an Assistant Professor with the University of Nebraska--Lincoln from 2016 to 2021.
He is currently an Associate Professor of Computer Science and the Director of the M.S. Software Engineering program at George Mason University, USA.
He is a recipient of the NSF CAREER Award, the NSF CRII Award, two Amazon Research Awards, and three test-of-time paper awards, including the IEEE TSE Most Influential Paper Award and the ACM SIGSOFT ICSE 10-Year Most Influential Paper Award.
He is a Senior Member of ACM and IEEE.
His research lies at the intersection of software engineering and formal methods, focusing on AI safety and program correctness.
\end{IEEEbiography}

\end{document}